\RequirePackage{fix-cm}
\documentclass[smallcondensed]{svjour3}     % onecolumn (ditto)
\smartqed  % flush right qed marks, e.g. at end of proof
\usepackage{graphicx}
\usepackage[utf8]{inputenc}
\usepackage{common}
\usepackage{mymacro}
\usepackage[hide]{notes} % \note conflict, use \notes instead
\allowdisplaybreaks
\newtheorem{observation}{Observation}

\begin{document}

\title{Ranking with submodular functions on a budget%\thanks{Grants or other notes
}


\author{Guangyi Zhang         \and
        Nikolaj Tatti	\and
        Aristides~Gionis %etc.
}


\institute{Guangyi Zhang \at
KTH Royal Institute of Technology, Sweden\\
              \email{guaz@kth.se} 
%              first address \\                            
%              Tel.: +123-45-678910\\
%              Fax: +123-45-678910\\
%              \email{fauthor@example.com}           %  \\
           \and
           Nikolaj Tatti \at
           HIIT, University of Helsinki, Finland\\
              \email{nikolaj.tatti@helsinki.fi}
           \and
           Aristides Gionis \at
           KTH Royal Institute of Technology, Sweden\\
              \email{argioni@kth.se}
}

\date{Received: date / Accepted: date}

\maketitle

\begin{abstract}
Submodular maximization has been the backbone of many important machine-learning problems, 
and has applications to viral marketing, diversification, sensor placement, and more.
However, the study of maximizing submodular functions has mainly been restricted 
in the context of selecting a set of items.
On the other hand, many real-world applications require a solution that is a ranking over a set of items.
The problem of ranking in the context of submodular function maximization has been considered before, 
but to a much lesser extent than item-selection formulations.
In this paper, we explore a novel formulation for ranking items with submodular valuations and budget constraints. 
We refer to this problem as  \emph{max-submodular ranking} (\msr).
In more detail, given a set of items and a set of non-decreasing submodular functions, 
where each function is associated with a budget, 
we aim to find a ranking of the set of items that maximizes 
the sum of values achieved by all functions under the budget constraints.
For the \msr problem with cardinality- and knapsack-type budget constraints
we propose practical algorithms with approximation guarantees.
In addition, we perform an empirical evaluation, 
which demonstrates the superior performance of the proposed algorithms
against strong baselines.
%\subclass{MSC code1 \and MSC code2 \and more}
\end{abstract}

\section{Introduction}
\label{section:intro}

Combinatorial optimization plays a central role in many machine-learning problems.
One prevalent approach to solve such problems is via \emph{submodular-optimization} techniques.
The popularity of submodular-optimization methods results from the fact that in many real-world settings 
the objective function exhibits the ``diminishing returns'' property, 
as well as from the ever-growing rich toolkit that has been developed in the past decades.
One fundamental primitive in this toolkit is \emph{submodular maximization} \citep{krause2014submodular}, 
which has been the backbone of a number of important problems, such as
sensor placement \citep{krause2008near},
viral marketing in social networks \citep{kempe2015maximizing},
document summarization \citep{lin2011class}, 
and more.

Submodular optimization has mainly been studied in the context of \emph{subset-selection problems}.
However, in many real-world applications the goal is to find a \emph{ranking} over a set of items.
Finding a ranking is a significantly more challenging task than subset selection,
as the search space is \revise{factorially} larger.
One successful attempt of applying ideas from submodular optimization to ranking 
is the \emph{submodular-ranking} problem (\sr) \citep{azar2011ranking}.
In this problem, given a set of items and a set of submodular functions, 
the goal is to find a (partial) ranking of the items so as to minimize the average ``cover time'' of all functions.

An exemplary application of \sr is in the \emph{multiple intents re-ranking} problem \citep{azar2009multiple}, 
which has applications in web searching.
In this problem setting, a user query may correspond to multiple user intents.
For example, a query of ``java'' may mean a programming language, an island, or a type of coffee.
Even for a seemingly unambiguous query, such as ``New York,'' 
there exist many possible intents, for example, attractions, cuisine, travel, cultural events, etc.
In the absence of an explicit user intent, we need to consider all possibilities.
The \sr formulation proposes to model each intent as a submodular function, 
whose value improves when a non-redundant web page of the right intent is encountered, 
and reaches a maximum when the user is satisfied, i.e., having gathered sufficient information.
The goal is to produce a ranking of web pages that minimizes the expected number of pages a user has to browse 
before they satisfy their information needs.
The expectation here is over the distribution of different user intents, 
which for this particular application can be assumed to be known.

While the \sr formulation can be useful in some cases, 
it fails to model realistically a number of other applications.
Critically, it assumes that a demand can wait indefinitely before it gets satisfied.
In the previous example, for instance, it is assumed that users will keep reading  
down a ranked list of web pages until they gather enough information.
In reality, a budget can be set for the amount of service that a user receives.
The budget can be the number of web pages to browse, 
or the time to spend on the web-search task.
%  or the number of places to visit in a journey.
A user stops receiving service once the budget is exceeded.
Moreover, the budget can vary across different demands.
For example, a user intent can be classified into one of three types, 
informational, navigational, and transactional \citep{jansen2008determining},
and each may come with a different budget, 
translating to the amount of ``patience'' that a user exhibit to obtain results for each type.
User intents and budgets can be readily extracted from the past search~logs.

To accommodate budgeted versions of the submodular ranking problem, 
we propose a new formulation, which we call \emph{max-submodular ranking} (\msr).
In the \msr problem, we are given a set of non-decreasing submodular functions, each associated with a budget.
We aim to find a ranking that, instead of minimizing the total coverage time of the functions, 
maximizes the sum of function values (coverage) under individual budget constraints.
In other words, every item in the ranking incurs a cost, and each function is evaluated 
at the maximal prefix of the ranked sequence that does not exceed its budget.
A precise formulation of the \msr problem is provided in Section~\ref{section:definition}.

In this paper, 
we propose practical algorithms with approximation guarantees for \msr, 
when the budget constraints are either cardinality or knapsack constraints.
We also note that the well-known \emph{constrained submodular maximization} and 
\emph{minimum submodular cover} problems are special cases of \msr and \sr, respectively, 
when there is a single submodular function.
In this sense, the \msr problem we define is a dual problem of \sr, 
in the same way that \emph{max $k$-cover} is a dual problem of \emph{minimum set cover}.

\msr has great potential to be applied in other scenarios, 
such as in the case where the submodular functions are 0--1 activation functions.
We call this special case \emph{max-activation ranking} (\mar) problem.
The idea is to activate as many demands as possible with a common ranking of items, or services,
under individual budget constraints.
As an example, 
some subscription-based streaming media services, such as Netflix, produce content in a data-driven fashion.
One possibility is to arrange the plot structure in a TV series such that the maximum number of audience 
will get interested before their individual cut-off points for a new show. 
The goal for the TV series producer is to encourage the maximum-size audience to continue watching.
A plot structure can be characterized as a sequence of scenes, each described by a set of tags, 
such as romantic, adventurous, funny, etc., which may interest particular audience.
Similar applications can also be found in ranking commercial ads,  ranking customer reviews, 
creating play lists for music streaming services, and more.

In concrete, our contributions in this paper are summarized as follows.
\begin{itemize}
	\item We introduce the novel problem of \emph{max-submodular ranking} (\msr), 
			where the goal is to find a ranking of a set of items
			so as to maximize the total value of a set of submodular functions
			under budget constraints.
	\item We prove that a simple greedy algorithm achieves a factor-2 approximation for the \msr problem 
		under cardinality constraints, which is tight for this particular greedy algorithm.
	\item We show that a weighted greedy algorithm that pays more attention to functions with small budget 
		achieves a factor-3 approximation for the \msr problem under cardinality constraints.
		While its worst-case bound is worse, there are natural problem instances 
		for which the weighted greedy finds better solutions than its unweighted counterpart.
	\item We devise a new algorithm that returns the best solution among the solutions found by 
		a cost-efficient greedy algorithm and a ranking of ``large'' items produced by dynamic programming. 
		Our algorithm achieves an approximation factor arbitrarily close to 4 for the \msr problem under knapsack constraints.
	\item We empirically evaluate and compare different algorithms on real-life datasets, and find that the proposed algorithms achieve superior performance when compared with strong baselines.
\end{itemize}

The rest of the paper is organized as follows.
We start by discussing the related work in Section \ref{section:related}, 
and we formally introduce the \msr problem in Section \ref{section:definition}.
The unweighted and weighted greedy algorithms for \msr under cardinality constraints are 
presented and analyzed in Sections \ref{section:greedy} and \ref{section:wgreedy}, respectively.
The novel algorithm for the \msr problem under knapsack constraints 
is introduced and analyzed in Section \ref{section:knapsack}.
We present our empirical evaluation in Section \ref{section:experiments}, 
and we offer our concluding remarks in Section \ref{section:conclusion}.

\section{Related work}
\label{section:related}

\paragraph{Submodular maximization.}
Submodular maximization is a special case of our formulation when given only a single function.
Coupled with a non-decreasing property and with a cardinality constraint
it is well-known that a simple greedy algorithm achieves a $e/(e-1)$ 
approximation \citep{nemhauser1978analysis}, which is also shown to be tight \citep{nemhauser1978best}. 
For a more general budget constraint, a natural algorithm is to return the best solution among the 
solutions found by a cost-efficient greedy method and by selecting the best singleton item.
Recently, the approximation factor of this ``best-of-two'' algorithm was shown to be within 
$[1/0.462, 1/0.427]$ \citep{feldman2020practical}.
A better 2-approximation is achieved by another greedy variant that returns the best solution
among the solutions found by a cost-efficient greedy algorithm and 
all its intermediate solutions, 
each augmented with the best single additional item \citep{yaroslavtsev2020bring}.

\paragraph{Submodularity for a sequence function.}
A sequential utility function is defined as $f: \seqs \to \real$, 
where \seqs is the set of all possible sequences of subsets of a ground set of items~\V.
Note that a set function can be seen as a special sequence function, in which the diminishing-returns effect holds for any subsequence relation.
\citet{streeter2008online} and \citet{zhang2012submodularity} introduce a notion of \emph{string submodularity}, which restricts the diminishing returns to only the prefix subsequence relation.
That is to say, a function $f$ is string submodular if appending an item to a sequence 
results in no larger marginal gain than appending the item to a prefix of the sequence.
The goal is to find a sequence of a given length that maximizes the value of the function $f$.
%To obtain solutions with approximation guarantees for this problem, 
%it is required that the function $f$ satisfies prefix-monotonicity and suffix-monotonicity properties, 
%in addition to string submodularity.
%Prefix-monotonicity states that a sequence is at least as good as its prefixes, 
%while suffix-monotonicity states that a sequence is at least as good as its suffixes.
In our formulation, the sum of multiple submodular functions remains submodular, and thus, string submodular.
However, the analysis in the prior work does not apply in our case 
as we assume that each submodular function is associated with a different budget constraint. 

% \notes[Guangyi]{Monotonicity is usually discussed without considering the budget.}

\paragraph{Submodular ranking.}
\citet{azar2011ranking} propose the \emph{submodular ranking} (\sr) problem, 
which aims to find a permutation to minimize the average ``cover time'' of a set of submodular functions, 
where we say that an input sequence ``covers'' a function if it evaluates to the maximum value of the function, 
and the ``cover time'' of a sequence of items is the shortest prefix of the sequence for which the function is covered.
The problem we study in this paper can be seen as a dual problem of the \sr problem.
The \sr problem originates from the classic min-sum set cover (\mssc) problem \citep{feige2004approximating} 
and its generalizations \citep{azar2009multiple,gamzu2010web}.

\paragraph{Diversified web search.}
In web search, in the absence of the explicit user intent, it is desirable to provide a sequence of high-quality and diverse documents that account for the interests of the overall user population.
Typically, the diversity is evaluated by the coverage at the topical level of some existing taxonomy \citep{zhai2015beyond}.
\citet{carbonell1998use} propose a greedy algorithm with respect to \emph{maximal marginal relevance} (MMR) to reduce the redundancy among returned documents.
\citet{bansal2010approximation} define the problem of finding an ordering of search results that maximizes the discounted cumulative gain (DCG), i.e., the sum of discounted gains of different user types, where the discount factor increases if a user type is satisfied later on.
They show that, in some special cases, the DCG metric can be rewritten as a weighted sum of submodular functions.
Our framework contributes to this theme by, for example, casting each user type or topic as a submodular function.

\section{Problem definition}
\label{section:definition}

We are given a universe set \V with $|\V|=\nV$ items, 
a set of \nfsms non-decreasing submodular functions $\fsms=\{ \fsm_1, \ldots, \fsm_\nfsms \}$, 
and a cost function $\cost: \V \to \real_+$.
Recall that a set function $\fsm: 2^\V \to \real_+$ is non-decreasing
if $\fsm(\sett) \le \fsm(\sets)$ for every $\sett\subseteq\sets\subseteq\V$, 
and it is submodular if 
$\fsm(\sett\cup\{\itemv\}) - \fsm(\sett)\ge\fsm(\sets\cup\{\itemv\}) - \fsm(\sets)$
for every $\sett\subseteq\sets\subseteq\V$ and $\itemv\in\V\setminus\sets$.
%The assumption $\fsm_i(\V)=1$ for every function $\fsm_i$ can be relaxed to $\fsm_i(\V) \in \real_+$ if the cost function $\cost$ is uniform.
%We will elaborate on this shortly.
Furthermore, each function $\fsm_i$ is associated with a budget $\bg_i \in \real_+$.
We will often write $\fsm(v \mid S)$ to mean $f(\{v\} \cup S) - f(S)$.

Let $\perm(\V)$ denote the set of permutations of \V, that is, 
$\perm(\V) =\{ \rk : \V \to \V \mid \rk \text{ is a permutation}\}$.
Our goal is to find a permutation $\rk \in \perm(\V)$ 
to maximize the sum of function values $\fsm_i(\rk_{\idx_i})$, 
where the input set $\rk_{\idx_i}$ is a prefix of the sought permutation \rk with feasibility constraints.
In particular, we consider that each function $\fsm_i$ receives as input 
the \emph{maximal prefix} of \rk that fits within its corresponding budget $\bg_i$.
\revise{In other words, the permutation \rk can be seen as a sequence of nested sets, one for each function.}
Formally, the \emph{max-submodular ranking} (\msr) problem that we study in this paper
is defined as follows.
\begin{problem}[Max-submodular ranking (\msr)]
\label{problem:msr}
Given a set of items \V, 
a set of non-decreasing and submodular functions $\fsms=\{ \fsm_1, \ldots, \fsm_\nfsms \}$, % with $\fsm_i(\V)=1$ for all $i \in [\nfsms]$,
a cost function $\cost: \V \to \real_+$, 
and non-negative budgets $\bg_i$ for each function $\fsm_i$, 
the \msr problem aims to find a permutation $\rk \in \perm(\V)$ that maximizes the sum
\begin{align}
 &   \sum_{\fsm_i \in \fsms} \fsm_i(\rk_{\idx_i}), \\
 \text{such that }~ &
    \idx_i = \max \{ j \in [\nV] : \cost(\rk_j) \le \bg_i \}, \nonumber 
\end{align}
where $\rk_j$ is the prefix of the permutation \rk of length $j$ and $\cost(\rk_j) = \sum_{v \in \rk_j} \cost(v)$. 
\end{problem}

We make a number of observations for Problem~\ref{problem:msr}.
%First, as mentioned before, if the cost function $\cost$ is uniform, i.e., $\cost(\cdot)=1$, without loss of generality we can assume
%that the maximum value of each function $\fsm_i$ is $\fsm_i(\V)=1$;
%otherwise we can scale the range of $\fsm_i$ appropriately, incorporate the scaling factor into a weight $\weight_i$, and maximize $\sum_{\fsm_i \in \fsms} \weight_i \fsm_i(\rk_{\idx_i})$ instead.
%Our techniques generalize easily to arbitrary non-negative weights $\weight_i$.
%However, when the cost function $\cost$ is non-uniform, our algorithm may need pseudo-polynomial time to finish, if we allow arbitrary non-negative weights.

Without loss of generality, we can assume that
$\fsm_i(\emptyset)=0$; 
otherwise we can translate the objective function by $\sum_{\fsm_i \in \fsms} \fsm_i(\emptyset)$.

%($ii$) $\cost: \V \to \naturalnum$ and $\bg_i \in \naturalnum$ 
%--- which can be achieved by scaling \cost and $\bg_i$ by a large constant.

Also note that not all items in the permutation solution \rk will necessarily be used
as an input to some function $\fsm_i\in\fsms$.
Instead, only the items in $\rk_{\idx_i}$ for the largest $\idx_i$ will be used.
For this reason, we can think that the output to the \msr problem is a \emph{partial} permutation;
after all functions deplete their budget, the remaining items of the permutation does not matter.

Finally, note that when the cost function \cost is uniform, i.e., $\cost(\cdot)=1$, \revise{we can consider only integral budget $\bg_i$ and assume $\idx_i = \bg_i$}.

% \notes[Aris]{Now \cost is modular. Does it make sense to consider a submodular (or supermodular) function?\\
% Guangyi: this becomes extremely hard even with just one \fsm. See 
% Iyer, R., \& Bilmes, J. (2013). Submodular Optimization with Submodular Cover and Submodular Knapsack Constraints.}

With respect to the hardness of approximation of the \msr problem, 
we observe that \msr is equivalent to the standard submodularity-maximization problem when $\nfsms=1$, 
that is, when there is only one function in \fsms. 
A second reduction from the standard submodularity-maximization problem
can be obtained by letting $\bg_i=\bg$, for all $i=1,\ldots,\nfsms$,
i.e., when the same budget is used for all functions.
The reason is that in this case the sum of submodular functions remains submodular, 
and we ask to maximize a submodular function under a cardinality constraint.
We conclude the following hardness result.

\begin{remark}[\cite{nemhauser1978best}]
For solving the max-submodular ranking (\msr) problem,
no algorithm requiring a polynomial number of function evaluations 
can achieve a better approximation guarantee than $e/(e-1)$.
\end{remark}

\revise{It is also well-known that \emph{maximum $k$-cover}, a special case of submodular maximization, is a dual problem to the \emph{minimum set cover} problem,
where the constraint in one problem is treated as the objective function in the other \citep{feige1998threshold}.
More generally, the \msr problem can be considered as the dual problem to the \emph{submodular-ranking} problem (\sr) \citep{azar2011ranking},
whose goal is to find a (partial) ranking of the items so as to minimize the average ``cover time'' of all functions.}

We conclude the section by 
introducing some additional notation that will be used in our analysis.
%We denote $\rsd_j$ to be the set of unsaturated functions before taking the $j$-th item, i.e., 
%$\rsd_j=\{ \fsm_i \in \fsms: \cost(\rk_{j-1}) < \bg_i \}$.
The optimal permutation is denoted by $\rkopt$.
We use the operator $\oplus$ to denote sequence concatenation and overload operator $\subseteq$ for subsequence relation.

\section{Cardinality constraints}
\label{section:cardinality}

We start our analysis of the \msr problem
for the case of cardinality constraints, that is, 
when the item costs are uniform ($\cost(\cdot)=1$).
For this particular case we present two algorithms, called \greedy and \wgreedy, 
both having provable guarantees.
Both algorithms generate a permutation by greedily selecting one item before the next.
Pseudocode for both algorithms is shown in a unified manner in Algorithm~\ref{alg}.
The difference in the two algorithms lies in adopting different coefficients $\coef_i$, 
associated with the submodular functions $\fsm_i$, in their selection criteria.
The first algorithm, \greedy, is an unweighted greedy ($\coef_i=1$) 
with respect to the submodular functions $\fsm_i$.
%that in each iteration picks the greedy item with respect to the objective function.
The second algorithm, \wgreedy, is a weighted greedy ($\coef_i=1/\bg_i$) 
that puts more weight on functions with smaller~budget.

The worst-case running time of both algorithms is $\bigO(\nV^2 \nfsms)$.
In practice, they run much faster and their actual running time grows almost linearly in $\nV$, 
thanks to applying a standard lazy evaluation technique \citep{leskovec2007cost}.
More details on scalability are discussed in Section \ref{subsection:runtime}.

%\noindent
%\rule{\columnwidth}{0.001mm}
\begin{algorithm}[t]
\caption{\greedyalg (A generalized algorithm for both \greedy and \wgreedy)}\label{alg}
\hspace*{\algorithmicindent} \textbf{Input:}
An instance of \msr and weights $\coef_i$, $i=1,\ldots,\nfsms$
%An instance $\inst = (\node, \clss, \tests, \cls, \prob, \cost, \thr)$, 
%a set of tests $\testset\subseteq\tests$ used so far,  impurity \\
%\hspace*{\algorithmicindent} function~\fimp, trade-off parameter $\paramfimp\ge 0$
%\\
%\hspace*{\algorithmicindent} \textbf{Output:} A decision tree \tree
\begin{algorithmic}[1]
\State $\rk \gets \emptyseq, j \gets 1$
\While{$j \le |\V|$}
 \State $\rsd_j \gets \{ \fsm_i \in \fsms: \cost(\rk) < \bg_i \}$ \Comment{ set of unsaturated functions}
 \State $v^* \gets \arg\max\limits_{v \in \V \setminus \pi} \left\{ \frac{1}{\cost(v)} \sum\limits_{\fsm_i \in \rsd_j: \cost(\rk)+\cost(v) \le \bg_i} \coef_{i} \fsm_i(v \mid \rk) \right\} $
 \Comment{ ties broken arbitrarily}
 \State $\rk \gets \rk \oplus v^*$  \Comment{ append $v^*$ at the end of sequence \rk}
 \State $j \gets j + 1$
\EndWhile
\State \textbf{Return} \rk
\end{algorithmic}
\end{algorithm}
%\noindent
%\rule{\columnwidth}{0.001mm} 

\subsection{Unweighted greedy}
\label{section:greedy}
We show that the unweighted greedy algorithm ($\coef_i=1$) 
achieves a 2-approximation guarantee for the \msr problem with uniform cost.
In addition, we show that the approximation ratio is tight for this particular algorithm.

\begin{theorem}
\label{thm:greedy}
\greedy $($Algorithm~\ref{alg} with coefficients $\coef_{i}=1$$)$ 
is a 2-approximation algorithm for the \msr problem with uniform item costs $(\cost(\cdot)=1)$.
\end{theorem}

\begin{proof}
Write $\rsd_j=\{ \fsm_i \in \fsms: \cost(\rk_{j-1}) < \bg_i \}$.
By the greedy selection criteria, 
we get that for arbitrary item $v \in \V$ in the $j$-th iteration it holds that
\begin{equation}
\label{eq:greedybest}
    \sum_{\fsm_i \in \rsd_{j}} \left( \fsm_i(\rk_{j}) - \fsm_i(\rk_{j-1}) \right) \ge 
    \sum_{\fsm_i \in \rsd_{j}} \fsm_i(v \mid \rk_{j-1}).
\end{equation}
The main idea of the proof is to choose an appropriate item $v$ 
for the above inequality at each iteration of the greedy, 
and sum over all iterations.
We denote the $j$-th item of the optimal permutation $\rkopt$ by $\vopt_j$. 
We write \ALG to denote the value achieved by the \greedy algorithm. 
Then
\begin{align*}
    \ALG & = \sum_{\fsm_i \in \fsms} \fsm_i(\rk_{\bg_i}) \\    
	& = \sum_{\fsm_i \in {\fsms}} \sum_{j=1}^{\bg_i} \left( \fsm_i(\rk_{j}) - \fsm_i(\rk_{j-1}) \right) && \triangleright\text{telescoping series} \\    
	& = \sum_{j=1}^{\nV} \sum_{\fsm_i \in \rsd_{j}} \left( \fsm_i(\rk_{j}) - \fsm_i(\rk_{j-1}) \right) &&  \\ 
    & \ge 
    \sum_{j=1}^{\nV} \sum_{\fsm_i \in \rsd_{j}} \fsm_i(\vopt_{j} \mid \rk_{j-1}) && \triangleright\text{Equation~(\ref{eq:greedybest})} \\    
	& = \sum_{\fsm_i \in {\fsms}} \sum_{j=1}^{\bg_i} \fsm_i(\vopt_{j} \mid \rk_{j-1}) \\    
    &\ge 
    \sum_{\fsm_i \in \fsms} \sum_{j=1}^{\bg_i} \fsm_i(\vopt_{j} \mid \rk_{\bg_i}) 
    && \triangleright\text{submodularity} \\    
    &\ge 
    \sum_{\fsm_i \in \fsms}
    \left( \fsm_i(\rkopt_{\bg_i} \cup \rk_{\bg_i}) - \fsm_i(\rk_{\bg_i}) \right)
    && \triangleright\text{submodularity} \\
    & \ge \sum_{\fsm_i \in \fsms}
    \left( \fsm_i(\rkopt_{\bg_i}) - \fsm_i(\rk_{\bg_i}) \right)
    && \triangleright\text{monotonicity} \\
	& = \OPT - \ALG.
 \end{align*}
Consequently, $2\ALG \geq \OPT$, proving the claim.\qed
\end{proof}

% In contract to the special case of maximizing a single submodular function, the lower bound for the greedy algorithm is larger for the \msr problem.
We complete the analysis of the \greedy algorithm for the \msr variant with cardinality constraints, 
by showing that the approximation ratio 2 is tight.

\begin{remark}
\label{remark:tight-2approx}

\greedy $($Algorithm~\ref{alg} with coefficients $\coef_{i}=1$$)$ 
cannot do better than 2-approximation 
for the \msr problem with uniform item costs $(\cost(\cdot)=1)$.
\end{remark}

\begin{proof}
We construct an instance where the algorithm returns $\ALG = \frac12 \OPT$.
The main idea is to force the algorithm to pick up items that are only beneficial to functions with large budget and ``starve'' those with small budget in the early iterations.
Consider functions $\fsm_i$ with budget $\bg_i=i$, for all $i \in [\nfsms]$.
Let $\nfsms=\nV$ be even, that is $\nfsms=\nV = 2k$ for some $k$.
Select $\epsilon > 0$.
For $i \leq k$, we define $\fsm_i(\rk) = \min\{1, I[v_i \in \rk] + \epsilon I[v_{i + k} \in \rk]\}$,
where $I[\cdot]$ is the indicator function.
For $i > k$, we define $\fsm_i(\rk) = I[v_i \in \rk]$.

Clearly every $\fsm_i$ is non-decreasing and submodular.
One possible optimal permutation is $\rkopt=(v_1,\ldots,v_\nV)$, which leads to $\OPT=m$.
Algorithm \ref{alg} with coefficient $\coef_i=1$ returns a permutation (out of many equivalent possible permutations) $\rk=(v_\nV,\ldots,v_1)$ with $\ALG=(1+\epsilon) m/2$.
By letting $\epsilon$ be arbitrarily small, we see that the bound in Theorem~\ref{thm:greedy} is tight.\qed
\end{proof}

\subsection{Weighted greedy}
\label{section:wgreedy}

Inspired by the instance that yields the tight bound in Remark \ref{remark:tight-2approx}, 
it is reasonable to let the algorithm favor functions with small budget at the early iterations.
Such a strategy is desirable as it in some sense suggests fairness in resource allocation, 
\revise{i.e., more functions can afford at least one item from the returned ranking}.
It also turns out to have better performance in experiments.
We show that such a strategy is indeed reliable by proving a constant-factor approximation guarantee.

\begin{theorem}
\label{thm:weighted-greedy}
\wgreedy $($Algorithm~\ref{alg} with coefficients $\coef_{i}=1/\bg_i$$)$ 
is a 3-approximation algorithm for the \msr problem with uniform item costs $(\cost(\cdot)=1)$.
\end{theorem}
\begin{proof}
Write $\rsd_j=\{ \fsm_i \in \fsms: \cost(\rk_{j-1}) < \bg_i \}$.
By the greedy selection criteria, 
we know that for an arbitrary item $v \in \V$ it holds that
\begin{equation}
\label{eq:greedyweighted}
    \sum_{\fsm_i \in \rsd_{j}} \coef_i (\fsm_i(\rk_{j}) - \fsm_i(\rk_{j-1})) \ge 
    \sum_{\fsm_i \in \rsd_{j}} \coef_i \fsm_i(v \mid \rk_{j-1}).
\end{equation}
We denote by $\vopt_j$ the $j$-th item of the optimal permutation \rkopt.
The idea is to replace the arbitrary item $v$ with $\vopt_k \in \rkopt$ and compute a weighted sum.
In order to define the weights,
given $k < j$, we write $\dup_{jk} = 1/2$, and $\dup_{jj} = (j + 1)/2$.
Immediately, $\sum_{k \in [j]} \dup_{jk} = (j - 1)/2 + (j + 1)/2 = j$.

Now Equation~(\ref{eq:greedyweighted}) implies
\begin{align*}
	\sum_{\fsm_i \in \fsms} \sum_{j \in [\bg_i]} j \coef_i (\fsm_i(\rk_{j}) - \fsm_i(\rk_{j-1})) 
	& = \sum_{j \in [\nV]} j \sum_{\fsm_i \in \rsd_{j}} \coef_i (\fsm_i(\rk_{j}) - \fsm_i(\rk_{j-1})) \\
	& = \sum_{j \in [\nV]} \sum_{k \in [j]} \dup_{jk} \sum_{\fsm_i \in \rsd_{j}} \coef_i (\fsm_i(\rk_{j}) - \fsm_i(\rk_{j-1})) \\
	& \ge \sum_{j \in [\nV]} \sum_{k \in [j]} \dup_{jk} \sum_{\fsm_i \in \rsd_{j}} \coef_i \fsm_i(\vopt_{k} \mid \rk_{j-1}). \\
    %& = \sum_{\fsm_i \in \fsms} \sum_{k \in [\nV]} \sum_{j \in [\bg_i]} \dup_{jk} \coef_i \fsm_i(\vopt_{k} \mid \rk_{j-1}).
\end{align*}

We will denote the left hand side of the above equation by LHS, and the right hand side by RHS.
We will first bound the RHS. 
In order to do so, we need an additional bound on the weights $\dup_{jk}$, namely, for any fixed $k$,
\begin{equation}
\label{eq:weights}
	\sum_{j = k}^{b} \frac{\dup_{jk}}{b} = \frac{k + 1}{2b} + \frac{b - k}{2b} = \frac{b + 1}{2b} > \frac{1}{2}.
\end{equation}

We can now bound the right hand side with
\begin{align*}
	\text{RHS} 
	&= \sum_{\fsm_i \in \fsms} \sum_{j \in [\bg_i]} \sum_{k \in [j]} \dup_{jk} \coef_i \fsm_i(\vopt_{k} \mid \rk_{j-1}) 
	\\
	&= \sum_{\fsm_i \in \fsms} \sum_{k \in [\bg_i]} \sum_{j = k}^{\bg_i} \dup_{jk} \coef_i \fsm_i(\vopt_{k} \mid \rk_{j-1}) 
	\\
	&\ge \sum_{\fsm_i \in \fsms} \sum_{k \in [\bg_i]}  \sum_{j = k}^{\bg_i} \dup_{jk} \coef_i \fsm_i(\vopt_{k} \mid \rk_{\bg_i}) 
	&& \triangleright\text{submodularity} \\
	&\ge \sum_{\fsm_i \in \fsms} \sum_{k \in [\bg_i]} \fsm_i(\vopt_{k} \mid \rk_{\bg_i}) /2 
	&& \triangleright\text{Equation~(\ref{eq:weights})}
	\allowdisplaybreaks
	\\
	&\ge \sum_{\fsm_i \in \fsms} (\fsm_i(\rkopt_{\bg_i} \cup \rk_{\bg_i}) - \fsm_i(\rk_{\bg_i})) /2 
	&& \triangleright\text{submodularity}
	\\
	&\ge \sum_{\fsm_i \in \fsms} (\fsm_i(\rkopt_{\bg_i}) - \fsm_i(\rk_{\bg_i})) /2 
	&& \triangleright\text{monotonicity}
	\\
	&= (\OPT - \ALG) /2.
\end{align*}

Now we consider the left hand side,
\begin{align*}
    \text{LHS} 
    %&= 
    %\sum_{\fsm_i \in \fsms} \sum_{j \in [\bg_i]}
    %\dup_j \coef_i (\fsm_i(\rk_{j}) - \fsm_i(\rk_{j-1}))
    %\\
    &=
    \sum_{\fsm_i \in \fsms} \sum_{j \in [\bg_i]}
    \frac{j}{\bg_i} (\fsm_i(\rk_{j}) - \fsm_i(\rk_{j-1}))
    \\
    &=
    \sum_{\fsm_i \in \fsms} \left(
    \frac{\bg_i}{\bg_i} \fsm_i(\rk_{\bg_i})
    - \sum_{j < \bg_i} \frac{j+1-j}{\bg_i} \fsm_i(\rk_{j}) \right)
    \allowdisplaybreaks
    \\
    %&=
    %\sum_{\fsm_i \in \fsms} \left(
    %\fsm_i(\rk_{\bg_i})
    %- \sum_{j < \bg_i} \frac{1}{\bg_i} \fsm_i(\rk_{j}) \right)
    %\\
    &\le 
    \sum_{\fsm_i \in \fsms} \fsm_i(\rk_{\bg_i})
	\\
	&= \ALG.
\end{align*}

Putting everything together,
$
\ALG \ge \text{LHS} \ge \text{RHS} \ge (\OPT - \ALG) /2,
$
and we obtain $3\ALG \ge \OPT$.\qed
\end{proof}

\iffalse
\begin{claim}
\label{claim:minimum-coef}
Given an \nV-by-\nV matrix with entry $\dup_{jk}$ specified as follows.
\begin{itemize}
	\item $\dup_{jk}=1/2, \forall k<j$, and $\dup_{jj}=j - (j-1)/2 = (j+1)/2$.
	\item $\dup_{jk}=0, \forall j<k$.
\end{itemize}
We have $\sum_{k \in [\nV]}\dup_{jk} = j$ for any $j \in [\nV]$ and 
$\min_{k \in [\bg]} \frac1{\bg} \sum_{j \in [\bg]} \dup_{jk} \ge 1/2$ for any integer $\bg \in [\nV]$.
\begin{proof}
The first identity is straightforward to verify.
$\sum_{k \in [\nV]}\dup_{jk} 
= \sum_{k \in [j]}\dup_{jk}
= \dup_{jj} + \sum_{k < j}\dup_{jk}
= (j+1)/2 + (j-1)/2
= j$.

Now we prove the second inequality.
We discuss different values of $k$.
\begin{itemize}
	\item When $k = \bg$,
	$
	\sum_{j \in [\bg]} \dup_{jk}  /\bg 
	= \dup_{\bg\bg}  /\bg 
	= \frac{\bg+1}{2} \frac{1}{\bg} > 1/2
	$.
	\item When $k < \bg$,
	$
	\sum_{j \in [\bg]} \dup_{jk}  /\bg 
	= (\dup_{kk} + \sum_{j > k}^{\bg} \dup_{jk})  /\bg 
	= (\frac{k+1}{2} + \frac{\bg-k}{2}) \frac{1}{\bg}
	= \frac{\bg+1}{2} \frac{1}{\bg} > 1/2
	$.
\end{itemize}
Hence, $\min_{k \in [\bg]} \frac1{\bg} \sum_{j \in [\bg]} \dup_{jk} \ge 1/2$ for any integer $\bg \in [\nV]$.
\end{proof}
\end{claim}
\fi

\section{Knapsack constraints}
\label{section:knapsack}

The traditional way of handling knapsack constraints is to adopt a 
cost-efficient variant of the greedy algorithm
where in each iteration we select the item with the largest ratio between utility and cost.
Furthermore, we compute a second solution by selecting the maximum-utility singleton item that is feasible.
The idea is to use the second solution to rescue the situation in which the greedy algorithm 
starts with some cost-efficient small items and then is ``starved'' 
(i.e., the remaining budget is not enough to admit another valuable large item).
This idea however falls short when it comes to the \msr problem.
The reason is that there are multiple knapsacks 
and each one of them may be ``starved'' by different big items.
A more sophisticated way is needed to compute an alternative second solution.

We now discuss our proposed method in more detail.
First, an item $v \in \V$ is called \emph{large} 
with respect to a function $\fsm_i \in \fsms$ 
if its cost is more than half of the budget~$\bg_i$, that is, $2\cost(v) > \bg_i$.
It is obvious that a function $\fsm_i$ can afford at most one large item.
The following variant of the \msr problem targets a similar objective to that of \msr, 
but exclusive to only large items.
\begin{problem}[Max-submodular ranking of large items (\msrl)]
\label{problem:msrl}
Given a set of items \V, 
a set of non-decreasing and submodular functions $\fsms=\{ \fsm_1, \ldots, \fsm_\nfsms \}$, %with $\fsm_i(\V)=1$ for all $i \in [\nfsms]$,
a cost function $\cost: \V \to \real_+$, 
and non-negative budgets $\bg_i$ for each function $\fsm_i$, 
the \msrl problem aims to find a permutation $\rk \in \perm(\V)$ that maximizes
\begin{align}
\fdp(\rk) = \sum_{v_j \in \rk} \fdp(v_j, \cost(\rk_{j-1}))
= \sum_{v_j \in \rk} \sum_{\fsm_i \in \fsms(v_j;\rk)} \fsm_i(v_j),
\label{eq:obj-dp}
\end{align}
where $\fsms(v_j;\rk)$ is the set of functions that take the $j$-th item $v_j \in \rk$ as a large item, i.e.,
$\fsms(v_j;\rk) = \{ \fsm_i \in \fsms: 2\cost(v_j) > \bg_i, \cost(\rk_{j}) \le \bg_i \}$,
and $\fdp(v_j, \cost)$ is defined to be the contribution of item $v_j$ by appending it to a prefix with cost $\cost$.
\end{problem}

We start by proving that the cost-efficient greedy algorithm yields a 3-approxima\-tion 
when there is no large item in \rkopt.
Next, we devise a dynamic programming (\DP) algorithm in Algorithm~\ref{alg-dp} to approximately solve \msrl.
Finally, we prove that the best solution among the greedy solution and the \DP solution 
can achieve an approximation guarantee that is arbitrarily close to~4.

\paragraph{Step 1: bounding small items in \rkopt.}
We first discuss the case in the absence of large items in \rkopt.
Let us introduce some notation.
We denote the $j$-th selected item by our algorithm by $u_j$.
We denote the $k$-th item of the optimal permutation $\rkopt$ by $\vopt_k$.
We denote the greedy solution of Algorithm \ref{alg} with coefficient $\coef_{i}=1$ by $\ALG_1$ and the \DP solution of Algorithm \ref{alg-dp} by $\ALG_2$.

The next theorem shows that, if every function $\fsm_i$ includes no such large item in \rkopt, $\ALG_1$ ensures a constant-factor guarantee.
Otherwise, we have an additional term $\fdp(\rkopt)$, which we will bound later.
%Note that if there are large items in \rkopt, then the following guarantee holds only against sequence $\rk'$, which is obtained by removing all large items in \rkopt.
\begin{theorem}
\label{theorem:knapsack-no-large-item}
The greedy algorithm yields $3\ALG_1 + \fdp(\rkopt) \ge \OPT$.
\end{theorem}

The proof relies on the next technical observation.
% that shows that, if item $\vopt_k$ is feasible and not large for function $\fsm_i$, then function $\fsm_i$ must be able to contribute to the greedy criterion evaluated at item $\vopt_k$ in early iterations.
%Note that $\ALG_1$ adopts a cost-efficient greedy criterion where a function $\fsm_i$ contributes only if the currently selected item together with previous selections does not violate its budget limit $\bg_i$.
\begin{observation}
\label{obs:feasible-vk}
For any $k$, if item $\vopt_k \in \rkopt$ is feasible and not large for function $\fsm_i$, i.e., $\cost(\rkopt_{k}) \le \bg_i$ and $2\cost(\vopt_k) \le \bg_i$, 
then at the $j$-th greedy iteration such that $\cost(\rk_{j-1}) \le \cost(\rkopt_{k})/2$, 
we have $\cost(\rk_{j-1}) + \cost(\vopt_k) \le \bg_i$.
\end{observation}
\begin{proof}
The proof is straightforward by combining 
$\cost(\rk_{j-1}) \le \cost(\rkopt_{k})/2 \le \bg_i/2$ and 
$\cost(\vopt_k) \le \bg_i/2$.
\qed
\end{proof}

\begin{proof}[of Theorem~\ref{theorem:knapsack-no-large-item}]
Write $\rsd_j=\{ \fsm_i \in \fsms: \cost(\rk_{j-1}) < \bg_i \}$.
By greedy, we know that for arbitrary item $v \in \V$ in the $j$-th iteration it holds that
\begin{align}
    \frac{1}{\cost(u_j)} \sum_{\fsm_i \in \rsd_{j}: \cost(\rk_{j}) \le \bg_i} \fsm_i(u_j \mid \rk_{j-1}) 
    &\ge 
    \frac{1}{\cost(v)} \sum_{\fsm_i \in \rsd_{j}: \cost(\rk_{j-1})+\cost(v) \le \bg_i} \fsm_i(v \mid \rk_{j-1}).
    \label{eq:knapsack-after-double-counting}
\end{align}

To simplify the notation used in the above inequality, let us define
$X_j = \{i \in [\nfsms] \mid \cost(\rk_{j}) \le \bg_i\}$
to be the valid function indices for $\rk_j$, and similarly
$Y_{jk} = \{i \in [\nfsms] \mid \cost(\rk_{j-1})+\cost(\vopt_k) \le \bg_i\}$.

For function $\fsm_i$, we define 
$\idxopt_i = \max \{ j \in [\nV] : \cost(\rkopt_j) \le \bg_i \}$. %, and 
%$\idxoptmax = \max_{i \in [\nfsms]} \idxopt_i$.
%We can safely assume $\cost(\rkopt) = \cost(\rkopt_{\idxoptmax})$ as further items do not improve the score.

Let us define a sequence of weights $\dup_j= \len(A_j)$, where the interval $A_j = (\cost(\rk_{j-1}), \cost(\rk_{j})] \cap (0, \cost(\rkopt)/2]$.

We will start by lower bounding $\ALG_1$ with
\begin{align*}
	\ALG_1
	&= \sum_{\fsm_i \in \fsms} \sum_{j \in [\idx_i]} \fsm_i(u_j \mid \rk_{j-1})  \\
	& = \sum_{j \in [\nV]} \sum_{i \in X_j} \fsm_i(u_j \mid \rk_{j-1})  &&  \\
    & \geq \sum_{j \in [\nV]} \frac{\dup_j}{\cost(u_j)} \sum_{i \in X_j} \fsm_i(u_j \mid \rk_{j-1}).
	&& \triangleright \text{since }\dup_j \le \cost(u_j) 
\end{align*}

Let us denote the right hand side with $C$. We will prove the theorem by showing that $C \geq (\OPT - \ALG_1 - \fdp(\rkopt)) /2$.

% In order to do so, 
We define
$\dup_{jk} = \len(A_j \cap B_k)$, where interval $B_k = (\cost(\rkopt_{k-1})/2, \cost(\rkopt_{k})/2]$.
We see immediately that
$\dup_{j} = \len(A_j) = \sum_{k \in [\nV]}\dup_{jk}$ as $B_k$ partition $A_j$.
Similarly, $\sum_{j \in [\nV]} \dup_{jk} = \len(B_k) = \cost(\vopt_k)/2$ as $A_j$ partition $B_k$.

We first claim that for any $i$,
\begin{equation}
	\text{if}\quad j > \idx_i \quad\text{and}\quad k \leq \idxopt_i,  \quad\text{then}\quad d_{jk} = 0.
\label{eq:claim1}
\end{equation}
To prove Equation~(\ref{eq:claim1}) note that $j - 1 \geq \idx_i$ implies that $\cost(\rk_{j - 1}) \geq b_i$
while $k \leq \idxopt_i$ implies that $\cost(\rkopt_{k}) \leq b_i$. Consequently, $A_j \cap B_k = \emptyset$ and $d_{jk} = 0$.

Let us now define $S_i = \left\{ k \in  [\idxopt_i] : 2\cost(\vopt_k) \leq b_i \right\}$ to be
the set of small items for the $i$-th function.
We claim that
\begin{equation}
	\text{if}\quad k \in S_i \quad\text{and}\quad d_{jk} > 0, \quad\text{then}\quad \cost(\rk_{j-1})+\cost(\vopt_k) \le \bg_i.
	\label{eq:claim2}
\end{equation}
To prove Equation~(\ref{eq:claim2}) note that
since $k \leq \idxopt_i$, we have $\cost(\rkopt_{k}) \le \bg_i$. Moreover, since $k \in S_i$,
we have $2\cost(\vopt_k) \le \bg_i$. If $\cost(\rk_{j-1}) > \cost(\rkopt_k) / 2$, then $A_j \cap B_k = \emptyset$ and so $d_{jk} = 0$.
Thus, $\cost(\rk_{j-1}) \leq \cost(\rkopt_k) / 2$. Observation~\ref{obs:feasible-vk} now proves Equation~(\ref{eq:claim2}).

We can now lower bound $C$ with
\begin{align}
    C & =
    \sum_{j \in [\nV]} \sum_{k \in [\nV]} \frac{\dup_{jk}}{\cost(u_j)} \sum_{i \in X_j} \fsm_i(u_j \mid \rk_{j-1}) 
    && \triangleright \text{since } \dup_{j} = \sum_{k \in [\nV]}\dup_{jk}
	\nonumber\\
    &\ge 
    \sum_{j \in [\nV]} \sum_{k \in [\nV]} \frac{\dup_{jk}}{\cost(\vopt_k)} \sum_{i \in Y_{jk}} \fsm_i(\vopt_k \mid \rk_{j-1}) 
    && \triangleright \text{Equation~(\ref{eq:knapsack-after-double-counting})}
    \nonumber \\
    & = \sum_{i \in [\nfsms]} \sum_{k \in [\nV]} \sum_{j \in [\idx_i]: i \in Y_{jk}} \frac{\dup_{jk}}{\cost(\vopt_k)} \fsm_i(\vopt_k \mid \rk_{j-1})  && \nonumber \\
    & \geq \sum_{i \in [\nfsms]} \sum_{k \in S_i} \sum_{j \in [\idx_i]: i \in Y_{jk}, \dup_{jk} > 0} \frac{\dup_{jk}}{\cost(\vopt_k)} \fsm_i(\vopt_k \mid \rk_{j-1})
    \nonumber \\
    %\sum_{k \in [\idxopt_i]} \sum_{j \in [\idx_i]} \frac{\dup_{jk}}{\cost(\vopt_k)} \fsm_i(\vopt_{k} \mid \rk_{j-1}) 
    & = \sum_{i \in [\nfsms]} \sum_{k \in S_i} \sum_{j \in [\idx_i]} \frac{\dup_{jk}}{\cost(\vopt_k)} \fsm_i(\vopt_k \mid \rk_{j-1})
    && \triangleright \text{Equation~(\ref{eq:claim2})}
    \nonumber \\
    &\ge
    \sum_{i \in [\nfsms]}\sum_{k \in S_i} \sum_{j \in [\idx_i]} \frac{\dup_{jk}}{\cost(\vopt_k)} \fsm_i(\vopt_{k} \mid \rk_{\idx_i})
    && \triangleright\text{submodularity}
    \nonumber \\
    & =
    \sum_{i \in [\nfsms]}\sum_{k \in S_i} \sum_{j \in [\nV]} \frac{\dup_{jk}}{\cost(\vopt_k)} \fsm_i(\vopt_{k} \mid \rk_{\idx_i})
    && \triangleright\text{Equation~(\ref{eq:claim1})}
    \nonumber \\
	&= \sum_{i \in [\nfsms]}\sum_{k \in S_i}\fsm_i(\vopt_{k} \mid \rk_{\idx_i}) /2
    && \triangleright \text{since } \sum_{j \in [\nV]} \dup_{jk} = \cost(\vopt_k)/2
    \allowdisplaybreaks
    \nonumber \\
	&\geq -\fdp(\rkopt)/2 + \sum_{i \in [\nfsms]}\sum_{k \in [\idxopt_i]}\fsm_i(\vopt_{k} \mid \rk_{\idx_i}) /2
    \allowdisplaybreaks
    \nonumber \\
    &\ge -\fdp(\rkopt)/2 + \sum_{i \in [\nfsms]} (\fsm_i(\rkopt_{\idxopt_i} \cup \rk_{\idx_i}) - \fsm_i(\rk_{\idx_i})) /2
    && \triangleright\text{submodularity}
    \allowdisplaybreaks
    \nonumber \\
    &\ge -\fdp(\rkopt)/2 + \sum_{i \in [\nfsms]} (\fsm_i(\rkopt_{\idxopt_i}) - \fsm_i(\rk_{\idx_i})) /2
    && \triangleright\text{monotonicity}
    \allowdisplaybreaks
    \nonumber \\
    &= (\OPT - \ALG_1 - \fdp(\rkopt)) /2.
    \nonumber 
\end{align}

Putting everything together, we obtain
$\ALG_1 \ge (\OPT - \ALG_1 - \fdp(\rkopt)) /2$, that is, 
$3\ALG_1 +  \fdp(\rkopt) \ge \OPT$.
\qed
\end{proof}

\paragraph{Step 2: bounding large items in \rkopt.}
\iffalse
When some functions do take large items in \OPT, we rely on the second solution.

A key observation is that, such large items must appear in \rkopt by an order of non-decreasing size.
\begin{observation}
\label{obs:non-decreasing-large-item}
Assume that two large items $v_i$ and $v_j$ appear in $\rk$ with $\cost(v_k) < \cost(v_{k'})$. Then

If two large items $v_k, v_{k'}$ appear in \rk such that $\cost(v_k) < \cost(v_{k'})$ and 
there exists some function $\fsm_i$ (resp. $\fsm_{i'}$) that satisfies 
$\cost(\rk_k) \le \bg_i$ and $2\cost(v_k) > \bg_i$
(resp. $\cost(\rk_{k'}) \le \bg_{i'}$ and $2\cost(v_{k'}) > \bg_{i'}$), then
$v_k$ is ranked before $v_{k'}$.
\end{observation}
\begin{proof}
We prove by contradiction.
If $v_k$ is ranked after $v_{k'}$, then $\bg_i \ge \cost(\rk_k) \ge \cost(v_k) + \cost(\rk_{k'}) \ge \cost(v_k) + \cost(v_{k'}) > 2\cost(v_k)$, 
which contradicts 
$2\cost(v_k) > \bg_i$.
\end{proof}
\fi

When some functions do take large items in \OPT, the quantity $\fdp(\rkopt)$ is positive, and we need to bound it.
We will do this by solving approximately the \msrl problem.

Our first result allows to order items based on their cost when solving \msrl.
\begin{theorem}
\label{theorem:monotone}
Assume a permutation $\rk$ with some item $v_i$ for which there is an index $j < i$
such that $\cost(v_j) \geq \cost(v_i)$.
Define a sub-permutation $\rk'$ by removing $v_i$.
Then $\fdp(\rk') \geq \fdp(\rk)$.
\end{theorem}

The proof relies on the following technical observation.
\begin{observation}
\label{obs:fsms}
Given an item $v$ and two sequences $\rk,\rk'$ with costs $\cost(\rk) \le \cost(\rk')$, we have
$\fsms(v;\rk' \concat v) \subseteq \fsms(v;\rk \concat v)$ and 
$\fdp(v; \cost(\rk)) \ge \fdp(v; \cost(\rk'))$.
\end{observation}
\begin{proof}
Note that
\begin{align*}
	\fsms(v;\rk \concat v) & = \{ \fsm_i \in \fsms: 2\cost(v) > \bg_i, \cost(\rk)+ \cost(v) \le \bg_i \} \\
	& \supseteq \{ \fsm_i \in \fsms: 2\cost(v) > \bg_i, \cost(\rk')+ \cost(v) \le \bg_i \} = \fsms(v;\rk' \concat v).
\end{align*}
Consequently, we have 
\[
\fdp(v; \cost(\rk)) = \sum_{\fsm_i \in \fsms(v;\rk \concat v)} \fsm_i(v) \ge \sum_{\fsm_i \in \fsms(v;\rk' \concat v)} \fsm_i(v) = \fdp(v; \cost(\rk')),
\]
proving the claim.\qed
\end{proof}

\begin{proof}[of Theorem~\ref{theorem:monotone}]
Let $v_i$ be an item that is in $\rk$ but not in $\rk'$.
Assume that $2\cost(v_i) > \bg$ for arbitrary function budget $\bg$. Then  $\cost(\rk_{i - 1}) + \cost(v_i) \geq 2 \cost(v_i) > \bg$, \revise{following the assumptions of the theorem}.
Consequently, $\fsms(v;\rk_i) = \emptyset$ and $\fdp(v_i, \cost(\rk_{i - 1})) = 0$.
Let $u_j$ be the $j$-th item in $\rk'$.  Observation~\ref{obs:fsms} now implies that
\[
	\fdp(\rk) = \sum_{v_i \in \rk} \fdp(v_i; \cost(\rk_{i-1})) = \sum_{v_i \in \rk'} \fdp(v_i; \cost(\rk_{i-1})) \leq \sum_{u_j \in \rk'} \fdp(u_j; \cost(\rk_{j-1}')) = \fdp(\rk'),
\]
proving the claim.\qed
\end{proof}

The above theorem enables a way to limit ourselves to sequences of large items with non-decreasing costs when solving \msrl.

Let us assume for simplicity that $\fdp(\cdot)$ is an integer-value in $[k]$.
We will discuss how to relax this assumption shortly.

We can solve \msrl by constructing a table \tb with entry $\tb(\val, j)$ for each value $\val \in [k]$ and each item with index $j \in [\nV]$.
We define the entry $\tb(\val, j)$ to be the lowest possible cost of a permutation using only the first $j$ items with at least value~\val,
\[
	\tb(\val, j)  = \min \{\cost(\rk) \mid \fdp(\rk) \geq \val, \pi \subseteq \revise{(}v_1, \ldots, v_j\revise{)} \}.
\]
Note that it is also possible to solve \msrl by defining a different dual \DP, where each entry $\tb(\bg, j)$ contains 
the highest value realizable by a permutation using only the first $j$ items with at most cost \bg.
However, this dual \DP is not amenable to the standard rounding trick we will introduce shortly. %, and thus remains a \PTAS.

\begin{theorem}
\label{theorem:dp}
The table $\tb$ satisfies the following relation:
\begin{equation}
	\tb(\val, j) = \min\left\{\tb(\val, j-1), \min_{\val' \mid \val' + \fdp({v_j}; \tb(\val', j - 1)) \ge \val} \tb(\val', j - 1) + \cost(v_j)\right\},
	\label{eq:dp}
\end{equation}
when $j > 1$. 
Moreover, $\tb(0, 1) = 0$, $\tb(\val, 1) = \cost(v_1)$ if $0 < \val \leq \fdp(v_1)$, and $\infty$ otherwise.
\end{theorem}
\begin{proof}
We will prove by induction.  The result holds trivially for $\tb(\val,1)$.

Next, we assume the theorem holds for all $\tb(\val', j - 1)$.
Now we examine $\tb(\val,j)$.
Let $\rk$ be a sequence responsible for $\tb(\val, j)$. Let $X$ be the value of the right hand side of Equation~\ref{eq:dp}.
\revise{Clearly, we have $X \geq \cost(\rk)$, and} we now prove the claim by showing that $X \leq \cost(\rk)$.

If $v_j$ not in \rk, then $X \leq \tb(\val, j - 1) \le \cost(\rk)$, and we are done.
If $v_j$ is in \rk, then let $\rk’$ be the permutation without $v_j$.
Let $\val' = \fdp(\rk’)$, and by the inductive hypothesis, we know that $T(\val', j - 1) \le \cost(\rk’)$. 
Then 
\[
	\val \le \fdp(\rk) = \val' + \fdp(v_j; \cost(\rk’)) \le \val' + \fdp(v_j; T(\val', j - 1)), 
\]
where the last inequality is by Observation \ref{obs:fsms}.
Therefore, according to the \DP updating rule, we have
\[
	X \le T(\val', j - 1) + \cost(v_j) \le \cost(\rk’) + \cost(v_j) = \cost(\rk),
\]
completing the proof.\qed
\end{proof}

We can use Theorem~\ref{theorem:dp} to construct $\tb$ using a dynamic program,
which is described in Algorithm \ref{alg-dp}.
Next, we will show that the \DP solves the \msrl problem.
%\noindent
%\rule{\columnwidth}{0.001mm}
\begin{algorithm}[t]
\caption{Dynamic program for solving \msrl}\label{alg-dp}
%\hspace*{\algorithmicindent} \textbf{Input:}
%An instance $\inst = (\node, \clss, \tests, \cls, \prob, \cost, \thr)$, 
%a set of tests $\testset\subseteq\tests$ used so far,  impurity \\
%\hspace*{\algorithmicindent} function~\fimp, trade-off parameter $\paramfimp\ge 0$
%\\
%\hspace*{\algorithmicindent} \textbf{Output:} A decision tree \tree

\begin{algorithmic}[1]
%\State Construct a table \tb with entry $\tb(\val, j)$ for each value $\val \in [0, \nfsms]$ and each item with index $j \in [\nV]$ s.t. $\cost(v_j) \le \cost(v_{j'})$ if $j \le j'$
%\State Let $\fdp(\cdot)$ be $\fdp(\rk) = \sum_{v_j \in \rk} \fdp_{v_j}(\cost(\rk_{j-1}))
%%= \sum_{v_j \in \rk} \sum_{\fsm_i \in \fsms: 2\cost(v_j) > \bg_i, \cost(\rk_j) \le \bg_i } \fsm_i(v_j)$
%\State Initiate $\tb(0, j) \gets \emptyseq$ for all $j$
%\State Initiate $\tb(\val, 1) \gets (v_1)$ for all $0 < \val \le \fdp_{v_1}(0)$, or $\tb(\val, 1) \gets \nullseq$ for $\val > \fdp_{v_1}(0)$
%\Comment{$\cost(\nullseq) := \infty$}

\State $\tb(\val, j) \gets \infty$ for all $\val$ and $j$
\State $\tb(\val, 1) \gets \cost(v_1)$ for all $0 < \val \le \fdp(v_1; 0)$, and $\tb(0, 1) \gets 0$

\For{$j = 2, \ldots, \nV$}
	%\State $B_j \gets \text{sort}(\{ \bg_i \mid i \in [\nfsms], 2\cost(v_j) > \bg_i \})$
	\For{$\val$ in descending order}
		\State $\val' \gets \val + \fdp(v_j, \tb(\val, j - 1))$
		\State $\tb(\val', j) \gets \min(\tb(\val', j),  \tb(\val, j - 1) + \cost(v_j))$
	\EndFor
	\State $x \gets \infty$
	\For{$\val$ in descending order}
		\State $\tb(\val, j) \gets \min(\tb(\val, j - 1), \, \tb(\val, j), \, x)$ 
		\State $x \gets \tb(\val, j)$
	\EndFor
\EndFor
\State \textbf{Return} Permutation responsible for $\tb(\val^*, \nV)$, 
where $\val^* = \max\{\val \mid \tb(\val, \nV) < \infty\} $
\end{algorithmic}
\end{algorithm}
%\noindent
%\rule{\columnwidth}{0.001mm} 

\begin{theorem}
\label{theorem:algdp}
Assume that $\fdp(\rk)$ is an integer in $[k]$ for every $\rk$.
\revise{The permutation \rk responsible for $\tb(\val^*, \nV)$, 
where $\val^* = \max\{\val \mid \tb(\val, \nV) < \infty\}$,
returned by Algorithm~\ref{alg-dp} has the largest $\fdp(\cdot)$ value.}
Besides, Algorithm~\ref{alg-dp} runs in $\bigO(n(k + \nfsms) + \nfsms \log \nfsms)$ time.
\end{theorem}

\begin{proof}
The correctness of the algorithm follows directly from Theorem~\ref{theorem:dp}.
There are in total $k \times \nV$ table entries.
Note that we can avoid directly invoking $\fdp(v_j; \cdot)$, which alone needs
time $\bigO(\nfsms)$, by sorting $\fsm_i$ by their budget $\bg_i$ and gradually
including more $\fsm_i$ as $\cost(\tb(\val, j - 1))$ and $\val$ decrease. This leads to an additional $\bigO(\nfsms)$ time per index $j$.
\qed
\end{proof}

\revise{We provide a numerical example to illustrate the \DP algorithm.
\begin{example}
\label{eg:dp}
Consider two modular functions $\fsm_1,\fsm_2$ with budget $\bg_1=3,\bg_2=9$, and 
three items $v_1,v_2,v_3$ with costs $2.5, 3, 6.5$, respectively.
We define 
$\fsm_1(v_1) = 1$,
$\fsm_1(v_2) = 1.5$,
$\fsm_2(v_3) = 1$, and 
0 otherwise.

It is easy to see that both the cost-efficient greedy algorithm and the best singleton will pick item $v_2$, which leads to a sub-optimal ranking, 
while the \DP algorithm can help us find the optimal ranking.

The \DP algorithm first initializes $\tb(\val, j) \gets \infty$ for all $\val$ and $j$.
We then process items $v_1,v_2,v_3$ in non-decreasing order by their costs.
\begin{itemize}
	\item Item $v_1$:
	we set $\tb(\val, 1) = \cost(v_1)$ for all $0 < \val \le \fsm_1(v_1)$ and $\tb(0, 1) = 0$.
	\item Item $v_2$:
	we set $\tb(\val, 2) = \tb(\val, 1)$ for all $\val \le \fsm_1(v_1)$, and
	$\tb(\val, 2) = \cost(v_2)$ for all $\fsm_1(v_1) < \val \le \fsm_1(v_2)$.
	\item Item $v_3$:
	we set $\tb(\val, 3) = \tb(\val, 2)$ for all $\val \le \fsm_1(v_2)$, and
	$\tb(\val, 3) = \cost(v_1) + \cost(v_3)$ for all $\fsm_1(v_2) < \val \le \fsm_1(v_1)+\fsm_2(v_3)$.
\end{itemize}
Finally, we return the permutation $\rk = (v_1,v_3)$ responsible for $\tb(\val^*, 3)$, 
where $\val^* = \fsm_1(v_1)+\fsm_2(v_3)$.
\end{example}}

So far we have assumed that $\fdp$ is an integer. Next,
we show that with a standard rounding technique, the \DP method in Algorithm \ref{alg-dp} gives an \FPTAS for \msrl.
The idea is to apply the \DP to a rounded instance, which is obtained by first scaling and rounding down every function $\lfloor \fsm_i / K \rfloor$
for certain $K$.

\begin{theorem}
\label{theorem:fptas}
Let $P = \max_{i, v} \fsm_i(v)$, where $v$ is a large item for $\fsm_i$.
Let $K = \frac{P\minval}{\nfsms}$ \revise{for any constant $\minval > 0$}. 
Define $\fsm'_i = \lfloor \fsm_i / K \rfloor$ and let $\fdp'(\rk)$ be the score
of a permutation using $\fsm'_i$ instead of $\fsm_i$.
Let $\rk$ be the permutation with the largest $\fdp(\rk)$.
Then $K\fdp'(\rk) \geq (1 - \minval)\fdp(\rk)$.
\end{theorem}

\begin{proof}
Due to scaling and rounding down we have $\fsm_i(v) - K \fsm'_i(v) \leq K$. 
Since there can be at most one large item per function,
and the score $\fdp$ contains at most $\nfsms$ functions,
thus, $\fdp(\rk) - K\fdp'(\rk) \leq \nfsms K = P \minval \leq \minval \fdp(\rk)$.
\qed
\end{proof}

\begin{corollary}
\label{corollary:dp}
Algorithm~\ref{alg-dp} with rounding yields $1/(1 - \minval)$ approximation guarantee
in $\bigO(n \nfsms^2 / \minval)$ time.
\end{corollary}

\begin{proof}
Let $\rk$ be the permutation with the largest $\fdp$ 
and let $\rk'$ be the permutation with the largest $\fdp'$. Then
$\fdp(\rk') \geq K\fdp'(\rk') \geq  K\fdp'(\rk) \geq (1 - \minval)\fdp(\rk)$, proving the approximation guarantee.

To prove the running time note that $\fdp(\cdot) \leq \nfsms P$ and $\fdp'(\cdot) \leq \nfsms P / K = m^2 / \minval$. Theorem~\ref{theorem:algdp} proves the claim.\qed
\end{proof}

We are finally ready to state our main result for \msr with non-uniform cost.
\begin{theorem}
\label{thm:knsapsack}
The best among Algorithm \ref{alg} with coefficient $\coef_{i}=1$ and Algorithm \ref{alg-dp} is $(3+1/(1-\minval))$-approximation for the \msr problem with non-uniform cost.
\end{theorem}
\begin{proof}

Theorem~\ref{theorem:knapsack-no-large-item} and Corollary~\ref{corollary:dp} imply that
\[
	(3 + (1 - \minval)^{-1})\ALG \geq 3 \ALG_1 + (1 - \minval)^{-1}\ALG_2 \geq \ALG_1 + \fdp(\rkopt) \geq \OPT,
\]
where $\ALG = \max \{ \ALG_1, \ALG_2 \}$,
proving the claim.\qed
\end{proof}

\iffalse
A key step (Step \ref{eq:require-no-large-item}) in Theorem \ref{theorem:knapsack-no-large-item} does not hold in the presence of large items in \rkopt.
Thus we need to further combine Lemma \ref{lemma:knapsack-no-large-item} and Corollary \ref{corollary:dp}.
By applying Lemma \ref{lemma:knapsack-no-large-item} to the sequence that is obtained by removing all large items in \rkopt, we know 
$$
\ALG_1 \ge \sum_{\fsm_i \in \fsms} \sum_{k \in [\idxopt_i]: 2\cost(\vopt_k) \le \bg_i} \fsm_i(\vopt_{k} \mid \rk_{\idx_i}) /2.
$$
By Corollary \ref{corollary:dp}, we know
\begin{align*}
	\frac{1}{2(1-\minval)} \ALG_2
	&\ge
	\sum_{\vopt_k \in \rkopt} \sum_{\fsm_i \in \fsms(\vopt_k;\rkopt)} \fsm_i(\vopt_k) /2
	\\
	&=
	\sum_{\fsm_i \in \fsms} \sum_{k \in [\idxopt_i]: 2\cost(\vopt_k) > \bg_i} \fsm_i(\vopt_{k}) /2
	\\
	&\ge
    \sum_{\fsm_i \in \fsms} \sum_{k \in [\idxopt_i]: 2\cost(\vopt_k) > \bg_i} \fsm_i(\vopt_{k} \mid \rk_{\idx_i}) /2.
\end{align*}
Therefore,
\begin{align*}
	\ALG_1 + \frac{1}{2(1-\minval)} \ALG_2 &\ge \sum_{\fsm_i \in \fsms} \sum_{k \in [\idxopt_i]} \fsm_i(\vopt_{k} \mid \rk_{\idx_i}) /2 \\
	&\ge \sum_{\fsm_i \in \fsms} (\fsm_i(\rkopt_{\idxopt_i} \cup \rk_{\idx_i}) - \fsm_i(\rk_{\idx_i})) /2 \\
	&\ge (\OPT - \ALG_1) /2, \\
	\shortintertext{and by rearranging,}
	3\ALG_1 + \ALG_2 / (1-\minval) &\ge \OPT \\
	(3+1/(1-\minval)) \ALG &\ge \OPT,
\end{align*}
where $\ALG = \max \{ \ALG_1, \ALG_2 \}$.
\end{proof}
\fi

\section{Experimental evaluation}
\label{section:experiments}

In this section, we evaluate the performance of the proposed algorithms on real-world datasets.
We first discuss our experimental evaluation for a playlist-making use-case. 
We model this use-case using the \emph{max-activation ranking} (\mar) problem, 
which is a special case of the \msr problem 
when the submodular functions $\fsm_i$ are 0--1 functions.
We then conduct two experiments for the \msr problem:
($i$) multiple intents re-ranking and 
($ii$) sequential active learning.
Finally, we evaluate the running time of our methods.
\revise{Statistics of the datasets used in the experiments are summarized in Table~\ref{tbl:datasets}.}
Our implementation and pre-processing scripts can be found in \code

\smallskip
\para{Proposed methods and baselines.}
The proposed greedy algorithms are denoted by \greedy and \wgreedy;
as discussed in Section~\ref{section:cardinality}.
The proposed dynamic program is denoted by \DP.
As baselines we use the following algorithms.
\begin{itemize}
	\item The greedy algorithm for the \sr problem~\citep{azar2011ranking},
	which favors functions near completion.
	We refer to this baseline as \greedysr.
	\item \revise{When only the minimum budget among all functions is considered, 
	the objective is a submodular function as a whole.
	We then consider the well-known ``best-of-two'' algorithm that returns the best solution among the 
	solutions found by a cost-efficient greedy method and by selecting the best singleton item.
	We refer to this baseline as \greedymin.}
	\item A simple ranking method (\quality) that orders individual items in non-increasing quality.
	\item A random ranking algorithm (\random).
\end{itemize}
\revise{Note that in general, computing the optimal solution requires enumerating all sequences of length equal to the maximum budget, 
which is computationally intractable even for a modest scenario with universe set $|\V|=100$ and budget $\bg=10$.}

\begin{table}[t]
  \caption{Datasets statistics}
  \label{tbl:datasets}
  \centering
\begin{tabular}{lrr}
\toprule
Dataset & $\nV = |\V|$ 
        & $\nfsms = |\fsms|$ \\
\midrule
Songs	&1872	&100	\\
Movies	&3669	&100	\\
Books	&3753	&1000	\\
20 Newsgroups	&172	&5	\\
Handwritten Digits	&1347	&3	\\
\bottomrule
\end{tabular}
\end{table}

\subsection{Experiments with the max-activation ranking (\mar) problem}

We evaluate our methods on three datasets, 
the Million Song dataset~\citep{Bertin-Mahieux2011}, 
the MovieLens dataset~\citep{harper2015movielens}, and
the Amazon Review dataset on books category~\citep{ni2019justifying}.
The three datasets have similar format, where
each record can be seen as a triple of user, item and rating.
We describe our experimental evaluation for the first dataset, and
the other two datasets are processed in the same way and give very similar results, 
as can be verified in Figure~\ref{fig:music}.

In the Million Song dataset, each record is a triple representing a user, song and play count.
We assume that a user likes a song if they play the song more than once.
We investigate an instance of the \mar problem for the 
application scenario of creating a \emph{playlist}.
In particular, we want to find a ranking of songs that maximizes 
the number of users who like at least one song among songs they listen to.
In this case, each user is modeled as a 0--1 activation function.
We generate a random budget for each user, 
i.e., the maximum number of songs a user will listen to, from 1 to a given maximum budget.
We also generate a random cost from 1 to 10 for each song 
in order to experiment with an additional non-uniform cost scenario.

The results of our evaluation are shown in Figure \ref{fig:music}.
The error bars are over random user budgets and item costs.
In the unit-cost scenario, the proposed \wgreedy algorithm is the best performing, 
closely followed by the proposed \greedy algorithm.
The performance of the baselines is inferior, 
and one reason is that they fail to take into account the user budget.
In the non-uniform cost scenario, the proposed \greedy algorithm obtains the best performance.
Note that it is expected that \DP has poor performance, as it is meant to help in extreme cases.
Also note that \DP does not scale for the book-list dataset ---
more details on scalability are discussed in Section \ref{subsection:runtime}.
Interestingly, \wgreedy performs worse than \greedysr, 
which indicates that a more sophisticated weighting scheme is needed to combine non-uniform budget and cost.

\begin{figure}[t]
    \centering
%    \captionsetup[subfigure]{position=b}
	\subcaptionbox{Songs (non-uniform\\ costs)}{
    % This file was created by tikzplotlib v0.9.8.
\begin{tikzpicture}

\definecolor{color0}{rgb}{0.12156862745098,0.466666666666667,0.705882352941177}
\definecolor{color1}{rgb}{1,0.498039215686275,0.0549019607843137}
\definecolor{color2}{rgb}{0.580392156862745,0.403921568627451,0.741176470588235}
\definecolor{color3}{rgb}{0.83921568627451,0.152941176470588,0.156862745098039}
\definecolor{color4}{rgb}{0.172549019607843,0.627450980392157,0.172549019607843}
\definecolor{color5}{rgb}{0.549019607843137,0.337254901960784,0.294117647058824}
\definecolor{color6}{rgb}{0.890196078431372,0.466666666666667,0.76078431372549}

\begin{axis}[width = 0.3\textwidth, height = 0.25\textwidth, legend image post style={scale=0.4}, ylabel style={align=center},
legend cell align={left},
legend style={
  fill opacity=0.8,
  draw opacity=1,
  text opacity=1,
  at={(0.5,0.91)},
  anchor=north,
  draw=white!80!black
},
tick align=outside,
tick pos=left,
x grid style={white!69.0196078431373!black},
xlabel={Maximum budget},
xmin=3.7, xmax=21.3,
xtick style={color=black},
y grid style={white!69.0196078431373!black},
ylabel={\#activated\\users},
ymin=-2.20109726779474, ymax=26.195864312413,
ytick style={color=black}
, legend style={at={(0.25,2.1)},anchor=north,legend columns=2,font=\scriptsize}]
\path [draw=color0, semithick]
(axis cs:5,5.00591415096638)
--(axis cs:5,9.43853029347807);

\path [draw=color0, semithick]
(axis cs:10,10.7311533266676)
--(axis cs:10,15.0466244511102);

\path [draw=color0, semithick]
(axis cs:15,17.1276147611929)
--(axis cs:15,21.0946074610293);

\path [draw=color0, semithick]
(axis cs:20,20.2060177796166)
--(axis cs:20,24.9050933314945);

\path [draw=color1, semithick]
(axis cs:5.5,4.40856765357792)
--(axis cs:5.5,8.92476567975541);

\path [draw=color1, semithick]
(axis cs:10.5,5.72022204589528)
--(axis cs:10.5,12.2797779541047);

\path [draw=color1, semithick]
(axis cs:15.5,10.1707398974507)
--(axis cs:15.5,15.1625934358826);

\path [draw=color1, semithick]
(axis cs:20.5,11.1032117153827)
--(axis cs:20.5,15.7856771735062);

\path [draw=color2, semithick]
(axis cs:5,3.99566814889469)
--(axis cs:5,7.11544296221642);

\path [draw=color2, semithick]
(axis cs:10,8.33552445496605)
--(axis cs:10,13.2200311005895);

\path [draw=color2, semithick]
(axis cs:15,14.5726216849363)
--(axis cs:15,18.5384894261748);

\path [draw=color2, semithick]
(axis cs:20,14.9859264662011)
--(axis cs:20,20.12518464491);

\path [draw=color3, semithick]
(axis cs:4.5,4.36164762243269)
--(axis cs:4.5,8.5272412664562);

\path [draw=color3, semithick]
(axis cs:9.5,8.13770037144415)
--(axis cs:9.5,12.3067440730003);

\path [draw=color3, semithick]
(axis cs:14.5,12.88855715957)
--(axis cs:14.5,17.7781095070967);

\path [draw=color3, semithick]
(axis cs:19.5,13.8711959445202)
--(axis cs:19.5,19.4621373888132);

\path [draw=color4, semithick]
(axis cs:5,-0.910326286876202)
--(axis cs:5,4.9103262868762);

\path [draw=color4, semithick]
(axis cs:5.2,1.19199156853911)
--(axis cs:5.2,5.919119542572);

\path [draw=color4, semithick]
(axis cs:10.2,1.78110653479845)
--(axis cs:10.2,9.33000457631266);

\path [draw=color4, semithick]
(axis cs:15.2,3.65340643997206)
--(axis cs:15.2,7.01326022669461);

\path [draw=color5, semithick]
(axis cs:5,0.0469610813936028)
--(axis cs:5,0.175261140828619);

\path [draw=color5, semithick]
(axis cs:10,-0.132905823079584)
--(axis cs:10,1.46623915641292);

\path [draw=color5, semithick]
(axis cs:10,0.0205606038414234)
--(axis cs:10,2.2016616183808);

\path [draw=color5, semithick]
(axis cs:15,0.266277342504992)
--(axis cs:15,2.17816710193945);

\path [draw=color6, semithick]
(axis cs:4.8,0.0803928301306396)
--(axis cs:4.8,2.80849605875825);

\path [draw=color6, semithick]
(axis cs:9.8,1.40049876048482)
--(axis cs:9.8,3.26616790618185);

\path [draw=color6, semithick]
(axis cs:14.8,2.79644344356249)
--(axis cs:14.8,5.20355655643751);

\path [draw=color6, semithick]
(axis cs:19.8,2.07495776931385)
--(axis cs:19.8,4.36948667513059);

\addplot [semithick, color0, mark=*, mark size=1, mark options={solid}]
table {%
5 7.22222222222222
10 12.8888888888889
15 19.1111111111111
20 22.5555555555556
};
\addlegendentry{Greedy-U}
\addplot [semithick, color1, mark=*, mark size=1, mark options={solid}]
table {%
5.5 6.66666666666667
10.5 9
15.5 12.6666666666667
20.5 13.4444444444444
};
\addlegendentry{Greedy-W}
\addplot [semithick, color2, mark=*, mark size=1, mark options={solid}]
table {%
5 5.55555555555556
10 10.7777777777778
15 16.5555555555556
20 17.5555555555556
};
\addlegendentry{Subm}
\addplot [semithick, color3, mark=*, mark size=1, mark options={solid}]
table {%
4.5 6.44444444444444
9.5 10.2222222222222
14.5 15.3333333333333
19.5 16.6666666666667
};
\addlegendentry{AG}
\addplot [semithick, color4, mark=*, mark size=1, mark options={solid}]
table {%
5.2 2
10.2 3.55555555555556
15.2 5.55555555555556
20.2 5.33333333333333
};
\addlegendentry{Quality}
\addplot [semithick, color5, mark=*, mark size=1, mark options={solid}]
table {%
5 0.111111111111111
10 0.666666666666667
15 1.11111111111111
20 1.22222222222222
};
\addlegendentry{Random}
\addplot [semithick, color6, mark=*, mark size=1, mark options={solid}]
table {%
4.8 1.44444444444445
9.8 2.33333333333333
14.8 4
19.8 3.22222222222222
};
\addlegendentry{DP}
\legend{}\end{axis}

\end{tikzpicture}
    }
    \hskip -2ex
    \subcaptionbox{Books (non-uniform\\ costs)}{
    % This file was created by tikzplotlib v0.9.8.
\begin{tikzpicture}

\definecolor{color0}{rgb}{0.12156862745098,0.466666666666667,0.705882352941177}
\definecolor{color1}{rgb}{1,0.498039215686275,0.0549019607843137}
\definecolor{color2}{rgb}{0.580392156862745,0.403921568627451,0.741176470588235}
\definecolor{color3}{rgb}{0.83921568627451,0.152941176470588,0.156862745098039}
\definecolor{color4}{rgb}{0.172549019607843,0.627450980392157,0.172549019607843}
\definecolor{color5}{rgb}{0.549019607843137,0.337254901960784,0.294117647058824}

\begin{axis}[width = 0.3\textwidth, height = 0.25\textwidth, legend image post style={scale=0.4}, ylabel style={align=center},
legend cell align={left},
legend style={
  fill opacity=0.8,
  draw opacity=1,
  text opacity=1,
  at={(0.03,0.97)},
  anchor=north west,
  draw=white!80!black
},
tick align=outside,
tick pos=left,
x grid style={white!69.0196078431373!black},
xlabel={Maximum budget},
xmin=3.7, xmax=21.3,
xtick style={color=black},
y grid style={white!69.0196078431373!black},
ylabel={\#activated\\users},
ymin=-2.47909376325106, ymax=28.7980211818341,
ytick style={color=black}
, legend style={at={(0.25,2.1)},anchor=north,legend columns=2,font=\scriptsize}]
\path [draw=color0, semithick]
(axis cs:5,4.55187099848267)
--(axis cs:5,8.114795668184);

\path [draw=color0, semithick]
(axis cs:10,9.84039042057622)
--(axis cs:10,16.8262762460904);

\path [draw=color0, semithick]
(axis cs:15,15.7274593830677)
--(axis cs:15,22.2725406169323);

\path [draw=color0, semithick]
(axis cs:20,23.0681103055688)
--(axis cs:20,27.3763341388757);

\path [draw=color1, semithick]
(axis cs:5.5,3.12587996719147)
--(axis cs:5.5,7.31856447725297);

\path [draw=color1, semithick]
(axis cs:10.5,6.0695841331371)
--(axis cs:10.5,13.2637492001962);

\path [draw=color1, semithick]
(axis cs:15.5,8.0504388360351)
--(axis cs:15.5,15.9495611639649);

\path [draw=color1, semithick]
(axis cs:20.5,8.3404462428423)
--(axis cs:20.5,17.8817759793799);

\path [draw=color2, semithick]
(axis cs:5,2.93648849964411)
--(axis cs:5,8.39684483368922);

\path [draw=color2, semithick]
(axis cs:10,7.92650465721482)
--(axis cs:10,12.7401620094518);

\path [draw=color2, semithick]
(axis cs:15,12.6917875065968)
--(axis cs:15,18.1971013822921);

\path [draw=color2, semithick]
(axis cs:20,15.1281893495545)
--(axis cs:20,21.7606995393344);

\path [draw=color3, semithick]
(axis cs:4.5,4.46309950887979)
--(axis cs:4.5,7.98134493556465);

\path [draw=color3, semithick]
(axis cs:9.5,7.88021275192923)
--(axis cs:9.5,13.2308983591819);

\path [draw=color3, semithick]
(axis cs:14.5,11.345120446603)
--(axis cs:14.5,17.5437684422858);

\path [draw=color3, semithick]
(axis cs:19.5,14.8879262479656)
--(axis cs:19.5,22.0009626409233);

\path [draw=color4, semithick]
(axis cs:5,-1.05740672029264)
--(axis cs:5,3.94629560918153);

\path [draw=color4, semithick]
(axis cs:10,-0.065023870382471)
--(axis cs:10,6.73169053704914);

\path [draw=color4, semithick]
(axis cs:5.2,1.05993519770101)
--(axis cs:5.2,9.38450924674343);

\path [draw=color4, semithick]
(axis cs:10.2,4.32672470481843)
--(axis cs:10.2,11.2288308507371);

\path [draw=color5, semithick]
(axis cs:5,0.0396961775744111)
--(axis cs:5,0.849192711314478);

\path [draw=color5, semithick]
(axis cs:10,-0.132905823079583)
--(axis cs:10,1.46623915641292);

\path [draw=color5, semithick]
(axis cs:15,-0.220361520099535)
--(axis cs:15,1.99813929787731);

\path [draw=color5, semithick]
(axis cs:10,0.258285947644837)
--(axis cs:10,2.18615849679961);

\addplot [semithick, color0, mark=*, mark size=1, mark options={solid}]
table {%
5 6.33333333333333
10 13.3333333333333
15 19
20 25.2222222222222
};
\addlegendentry{Greedy-U}
\addplot [semithick, color1, mark=*, mark size=1, mark options={solid}]
table {%
5.5 5.22222222222222
10.5 9.66666666666667
15.5 12
20.5 13.1111111111111
};
\addlegendentry{Greedy-W}
\addplot [semithick, color2, mark=*, mark size=1, mark options={solid}]
table {%
5 5.66666666666667
10 10.3333333333333
15 15.4444444444444
20 18.4444444444444
};
\addlegendentry{Subm}
\addplot [semithick, color3, mark=*, mark size=1, mark options={solid}]
table {%
4.5 6.22222222222222
9.5 10.5555555555556
14.5 14.4444444444444
19.5 18.4444444444444
};
\addlegendentry{AG}
\addplot [semithick, color4, mark=*, mark size=1, mark options={solid}]
table {%
5.2 1.44444444444445
10.2 3.33333333333334
15.2 5.22222222222222
20.2 7.77777777777778
};
\addlegendentry{Quality}
\addplot [semithick, color5, mark=*, mark size=1, mark options={solid}]
table {%
5 0.444444444444444
10 0.666666666666667
15 0.888888888888889
20 1.22222222222222
};
\addlegendentry{Random}
\legend{}\end{axis}

\end{tikzpicture}
    }
    \hskip -2ex
    \subcaptionbox{Movies (non-uniform\\ costs)}{
    % This file was created by tikzplotlib v0.9.8.
\begin{tikzpicture}

\definecolor{color0}{rgb}{0.12156862745098,0.466666666666667,0.705882352941177}
\definecolor{color1}{rgb}{1,0.498039215686275,0.0549019607843137}
\definecolor{color2}{rgb}{0.580392156862745,0.403921568627451,0.741176470588235}
\definecolor{color3}{rgb}{0.83921568627451,0.152941176470588,0.156862745098039}
\definecolor{color4}{rgb}{0.172549019607843,0.627450980392157,0.172549019607843}
\definecolor{color5}{rgb}{0.549019607843137,0.337254901960784,0.294117647058824}
\definecolor{color6}{rgb}{0.890196078431372,0.466666666666667,0.76078431372549}

\begin{axis}[width = 0.3\textwidth, height = 0.25\textwidth, legend image post style={scale=0.4}, ylabel style={align=center},
legend cell align={left},
legend style={
  fill opacity=0.8,
  draw opacity=1,
  text opacity=1,
  at={(0.97,0.03)},
  anchor=south east,
  draw=white!80!black
},
tick align=outside,
tick pos=left,
x grid style={white!69.0196078431373!black},
xlabel={Maximum budget},
xmin=3.7, xmax=21.3,
xtick style={color=black},
y grid style={white!69.0196078431373!black},
ylabel={\#activated\\users},
ymin=-4.16256049228811, ymax=89.5676506686618,
ytick style={color=black}
, legend style={at={(0.25,2.1)},anchor=north,legend columns=2,font=\scriptsize}]
\path [draw=color0, semithick]
(axis cs:5,37.0462221401567)
--(axis cs:5,52.9537778598433);

\path [draw=color0, semithick]
(axis cs:10,57.1279543522369)
--(axis cs:10,66.4276012033187);

\path [draw=color0, semithick]
(axis cs:15,71.7770580201237)
--(axis cs:15,79.1118308687652);

\path [draw=color0, semithick]
(axis cs:20,74.6462292381195)
--(axis cs:20,84.4648818729916);

\path [draw=color1, semithick]
(axis cs:5.5,34.7958817330258)
--(axis cs:5.5,50.7596738225298);

\path [draw=color1, semithick]
(axis cs:10.5,55.051863697187)
--(axis cs:10.5,66.7259140805908);

\path [draw=color1, semithick]
(axis cs:15.5,72.7263067632925)
--(axis cs:15.5,79.4959154589297);

\path [draw=color1, semithick]
(axis cs:20.5,75.5817023639067)
--(axis cs:20.5,85.3071865249822);

\path [draw=color2, semithick]
(axis cs:5,35.1172690910064)
--(axis cs:5,51.3271753534381);

\path [draw=color2, semithick]
(axis cs:10,54.5219209071794)
--(axis cs:10,63.7003013150428);

\path [draw=color2, semithick]
(axis cs:15,69.8132052472746)
--(axis cs:15,78.853461419392);

\path [draw=color2, semithick]
(axis cs:20,72.4773835268811)
--(axis cs:20,83.3003942508967);

\path [draw=color3, semithick]
(axis cs:4.5,35.1172690910064)
--(axis cs:4.5,51.3271753534381);

\path [draw=color3, semithick]
(axis cs:9.5,54.5219209071794)
--(axis cs:9.5,63.7003013150428);

\path [draw=color3, semithick]
(axis cs:14.5,69.8132052472746)
--(axis cs:14.5,78.853461419392);

\path [draw=color3, semithick]
(axis cs:19.5,72.4773835268811)
--(axis cs:19.5,83.3003942508967);

\path [draw=color4, semithick]
(axis cs:5.2,0.322211606199845)
--(axis cs:5.2,27.0111217271335);

\path [draw=color4, semithick]
(axis cs:10.2,6.22418090257929)
--(axis cs:10.2,39.3313746529763);

\path [draw=color4, semithick]
(axis cs:15.2,28.0247204226096)
--(axis cs:15.2,59.7530573551682);

\path [draw=color4, semithick]
(axis cs:20.2,29.622184494131)
--(axis cs:20.2,56.6000377280913);

\path [draw=color5, semithick]
(axis cs:5,0.151568731270138)
--(axis cs:5,3.84843126872986);

\path [draw=color5, semithick]
(axis cs:10,0.10916448691833)
--(axis cs:10,2.33527995752611);

\path [draw=color5, semithick]
(axis cs:15,0.0979036513914338)
--(axis cs:15,8.79098523749746);

\path [draw=color5, semithick]
(axis cs:20,0.334579130168352)
--(axis cs:20,5.66542086983165);

\path [draw=color6, semithick]
(axis cs:4.8,15.7600766025731)
--(axis cs:4.8,31.7954789529824);

\path [draw=color6, semithick]
(axis cs:9.8,28.6007843592995)
--(axis cs:9.8,44.5103267518116);

\path [draw=color6, semithick]
(axis cs:14.8,27.9623481661201)
--(axis cs:14.8,46.4820962783243);

\path [draw=color6, semithick]
(axis cs:19.8,31.1414619797326)
--(axis cs:19.8,42.6363157980452);

\addplot [semithick, color0, mark=*, mark size=1, mark options={solid}]
table {%
5 45
10 61.7777777777778
15 75.4444444444444
20 79.5555555555556
};
\addlegendentry{Greedy-U}
\addplot [semithick, color1, mark=*, mark size=1, mark options={solid}]
table {%
5.5 42.7777777777778
10.5 60.8888888888889
15.5 76.1111111111111
20.5 80.4444444444444
};
\addlegendentry{Greedy-W}
\addplot [semithick, color2, mark=*, mark size=1, mark options={solid}]
table {%
5 43.2222222222222
10 59.1111111111111
15 74.3333333333333
20 77.8888888888889
};
\addlegendentry{Subm}
\addplot [semithick, color3, mark=*, mark size=1, mark options={solid}]
table {%
4.5 43.2222222222222
9.5 59.1111111111111
14.5 74.3333333333333
19.5 77.8888888888889
};
\addlegendentry{AG}
\addplot [semithick, color4, mark=*, mark size=1, mark options={solid}]
table {%
5.2 13.6666666666667
10.2 22.7777777777778
15.2 43.8888888888889
20.2 43.1111111111111
};
\addlegendentry{Quality}
\addplot [semithick, color5, mark=*, mark size=1, mark options={solid}]
table {%
5 2
10 1.22222222222222
15 4.44444444444444
20 3
};
\addlegendentry{Random}
\addplot [semithick, color6, mark=*, mark size=1, mark options={solid}]
table {%
4.8 23.7777777777778
9.8 36.5555555555556
14.8 37.2222222222222
19.8 36.8888888888889
};
\addlegendentry{DP}
%\legend{}
\end{axis}

\end{tikzpicture}
    }
    \subcaptionbox{Songs (unit costs)}{
    % This file was created by tikzplotlib v0.9.8.
\begin{tikzpicture}

\definecolor{color0}{rgb}{0.12156862745098,0.466666666666667,0.705882352941177}
\definecolor{color1}{rgb}{1,0.498039215686275,0.0549019607843137}
\definecolor{color2}{rgb}{0.580392156862745,0.403921568627451,0.741176470588235}
\definecolor{color3}{rgb}{0.83921568627451,0.152941176470588,0.156862745098039}
\definecolor{color4}{rgb}{0.172549019607843,0.627450980392157,0.172549019607843}
\definecolor{color5}{rgb}{0.549019607843137,0.337254901960784,0.294117647058824}

\begin{axis}[width = 0.3\textwidth, height = 0.25\textwidth, legend image post style={scale=0.4}, ylabel style={align=center},
legend cell align={left},
legend style={
  fill opacity=0.8,
  draw opacity=1,
  text opacity=1,
  at={(0.5,0.91)},
  anchor=north,
  draw=white!80!black
},
tick align=outside,
tick pos=left,
x grid style={white!69.0196078431373!black},
xlabel={Maximum budget},
xmin=3.7, xmax=21.3,
xtick style={color=black},
y grid style={white!69.0196078431373!black},
ylabel={\#activated\\users},
ymin=1.53710594091525, ymax=47.1937558826023,
ytick style={color=black}
, legend style={at={(0.25,2.1)},anchor=north,legend columns=2,font=\scriptsize}]
\path [draw=color0, semithick]
(axis cs:5,16.5665421941156)
--(axis cs:5,18.3223466947732);

\path [draw=color0, semithick]
(axis cs:10,24.2194083332312)
--(axis cs:10,27.7805916667688);

\path [draw=color0, semithick]
(axis cs:15,32.6701794426022)
--(axis cs:15,35.9964872240645);

\path [draw=color0, semithick]
(axis cs:20,37.366055144207)
--(axis cs:20,41.7450559669041);

\path [draw=color1, semithick]
(axis cs:5,16.70050787707)
--(axis cs:5,18.1883810118189);

\path [draw=color1, semithick]
(axis cs:10,25.6445913674898)
--(axis cs:10,29.0220752991768);

\path [draw=color1, semithick]
(axis cs:15,34.8357553774246)
--(axis cs:15,37.830911289242);

\path [draw=color1, semithick]
(axis cs:20,42.4371019430299)
--(axis cs:20,45.1184536125256);

\path [draw=color2, semithick]
(axis cs:5.5,13.532110683054)
--(axis cs:5.5,15.5790004280571);

\path [draw=color2, semithick]
(axis cs:10.5,19.2131959619264)
--(axis cs:10.5,23.2312484825181);

\path [draw=color2, semithick]
(axis cs:15.5,25.5768175366273)
--(axis cs:15.5,28.4231824633727);

\path [draw=color2, semithick]
(axis cs:20.5,29.9032103844096)
--(axis cs:20.5,33.8745673933682);

\path [draw=color3, semithick]
(axis cs:4.5,12.0583986819747)
--(axis cs:4.5,15.4971568735808);

\path [draw=color3, semithick]
(axis cs:9.5,18.9790763352433)
--(axis cs:9.5,23.0209236647567);

\path [draw=color3, semithick]
(axis cs:14.5,25.0478883621024)
--(axis cs:14.5,28.9521116378976);

\path [draw=color3, semithick]
(axis cs:19.5,29.1193236977841)
--(axis cs:19.5,35.3251207466603);

\path [draw=color4, semithick]
(axis cs:5,11.594521225502)
--(axis cs:5,13.7388121078313);

\path [draw=color4, semithick]
(axis cs:10,17.1066285106884)
--(axis cs:10,20.0044826004227);

\path [draw=color4, semithick]
(axis cs:15,20.9182189403113)
--(axis cs:15,23.7484477263554);

\path [draw=color4, semithick]
(axis cs:20,23.3533223009129)
--(axis cs:20,26.6466776990871);

\path [draw=color5, semithick]
(axis cs:5,3.61240821099193)
--(axis cs:5,4.38759178900807);

\path [draw=color5, semithick]
(axis cs:10,5.34847231442767)
--(axis cs:10,6.42930546335011);

\path [draw=color5, semithick]
(axis cs:15,7.12120906343717)
--(axis cs:15,10.6565687143406);

\path [draw=color5, semithick]
(axis cs:20,7.98410855328217)
--(axis cs:20,11.1270025578289);

\addplot [semithick, color0, mark=*, mark size=1, mark options={solid}]
table {%
5 17.4444444444444
10 26
15 34.3333333333333
20 39.5555555555556
};
\addlegendentry{Greedy-U}
\addplot [semithick, color1, mark=*, mark size=1, mark options={solid}]
table {%
5 17.4444444444444
10 27.3333333333333
15 36.3333333333333
20 43.7777777777778
};
\addlegendentry{Greedy-W}
\addplot [semithick, color2, mark=*, mark size=1, mark options={solid}]
table {%
5.5 14.5555555555556
10.5 21.2222222222222
15.5 27
20.5 31.8888888888889
};
\addlegendentry{Subm}
\addplot [semithick, color3, mark=*, mark size=1, mark options={solid}]
table {%
4.5 13.7777777777778
9.5 21
14.5 27
19.5 32.2222222222222
};
\addlegendentry{AG}
\addplot [semithick, color4, mark=*, mark size=1, mark options={solid}]
table {%
5 12.6666666666667
10 18.5555555555556
15 22.3333333333333
20 25
};
\addlegendentry{Quality}
\addplot [semithick, color5, mark=*, mark size=1, mark options={solid}]
table {%
5 4
10 5.88888888888889
15 8.88888888888889
20 9.55555555555556
};
\addlegendentry{Random}
\legend{}\end{axis}

\end{tikzpicture}
    }
    \hskip -2ex
    \subcaptionbox{Books (unit costs)}{
    % This file was created by tikzplotlib v0.9.8.
\begin{tikzpicture}

\definecolor{color0}{rgb}{0.12156862745098,0.466666666666667,0.705882352941177}
\definecolor{color1}{rgb}{1,0.498039215686275,0.0549019607843137}
\definecolor{color2}{rgb}{0.580392156862745,0.403921568627451,0.741176470588235}
\definecolor{color3}{rgb}{0.83921568627451,0.152941176470588,0.156862745098039}
\definecolor{color4}{rgb}{0.172549019607843,0.627450980392157,0.172549019607843}
\definecolor{color5}{rgb}{0.549019607843137,0.337254901960784,0.294117647058824}

\begin{axis}[width = 0.3\textwidth, height = 0.25\textwidth, legend image post style={scale=0.4}, ylabel style={align=center},
legend cell align={left},
legend style={
  fill opacity=0.8,
  draw opacity=1,
  text opacity=1,
  at={(0.5,0.91)},
  anchor=north,
  draw=white!80!black
},
tick align=outside,
tick pos=left,
x grid style={white!69.0196078431373!black},
xlabel={Maximum budget},
xmin=3.7, xmax=21.3,
xtick style={color=black},
y grid style={white!69.0196078431373!black},
ylabel={\#activated\\users},
ymin=0.35357699892469, ymax=52.1480349447538,
ytick style={color=black}
, legend style={at={(0.25,2.1)},anchor=north,legend columns=2,font=\scriptsize}]
\path [draw=color0, semithick]
(axis cs:5,17.2918309225527)
--(axis cs:5,18.7081690774473);

\path [draw=color0, semithick]
(axis cs:10,28.0037221509024)
--(axis cs:10,29.1073889602087);

\path [draw=color0, semithick]
(axis cs:15,35.344225229504)
--(axis cs:15,39.3224414371627);

\path [draw=color0, semithick]
(axis cs:20,45.2792775886254)
--(axis cs:20,48.4985001891524);

\path [draw=color1, semithick]
(axis cs:5,17.2604712745321)
--(axis cs:5,18.9617509476901);

\path [draw=color1, semithick]
(axis cs:10,28.1382357430811)
--(axis cs:10,30.5284309235855);

\path [draw=color1, semithick]
(axis cs:15,38.0242015483769)
--(axis cs:15,40.4202428960675);

\path [draw=color1, semithick]
(axis cs:20,47.9840363760162)
--(axis cs:20,49.7937414017615);

\path [draw=color2, semithick]
(axis cs:5.5,13.4998778920984)
--(axis cs:5.5,16.2778998856794);

\path [draw=color2, semithick]
(axis cs:10.5,19.6201851295586)
--(axis cs:10.5,23.4909259815525);

\path [draw=color2, semithick]
(axis cs:15.5,27.1683385942629)
--(axis cs:15.5,31.0538836279594);

\path [draw=color2, semithick]
(axis cs:20.5,30.9126202955638)
--(axis cs:20.5,37.3096019266584);

\path [draw=color3, semithick]
(axis cs:4.5,13.2682775810324)
--(axis cs:4.5,15.3983890856343);

\path [draw=color3, semithick]
(axis cs:9.5,19.7001236947393)
--(axis cs:9.5,23.8554318608163);

\path [draw=color3, semithick]
(axis cs:14.5,25.4997145299112)
--(axis cs:14.5,31.1669521367554);

\path [draw=color3, semithick]
(axis cs:19.5,30.8442714338316)
--(axis cs:19.5,37.1557285661684);

\path [draw=color4, semithick]
(axis cs:5,13.1638449366498)
--(axis cs:5,16.8361550633502);

\path [draw=color4, semithick]
(axis cs:10,17.8325842691864)
--(axis cs:10,21.0563046197025);

\path [draw=color4, semithick]
(axis cs:15,23.9955934498838)
--(axis cs:15,29.1155176612273);

\path [draw=color4, semithick]
(axis cs:20,29.7936602918075)
--(axis cs:20,33.9841174859703);

\path [draw=color5, semithick]
(axis cs:5,2.70787054191692)
--(axis cs:5,3.95879612474975);

\path [draw=color5, semithick]
(axis cs:10,3.93640885604285)
--(axis cs:10,6.28581336617937);

\path [draw=color5, semithick]
(axis cs:15,5.36674481091094)
--(axis cs:15,9.74436630020017);

\path [draw=color5, semithick]
(axis cs:20,8.13949802285743)
--(axis cs:20,11.6382797549203);

\addplot [semithick, color0, mark=*, mark size=1, mark options={solid}]
table {%
5 18
10 28.5555555555556
15 37.3333333333333
20 46.8888888888889
};
\addlegendentry{Greedy-U}
\addplot [semithick, color1, mark=*, mark size=1, mark options={solid}]
table {%
5 18.1111111111111
10 29.3333333333333
15 39.2222222222222
20 48.8888888888889
};
\addlegendentry{Greedy-W}
\addplot [semithick, color2, mark=*, mark size=1, mark options={solid}]
table {%
5.5 14.8888888888889
10.5 21.5555555555556
15.5 29.1111111111111
20.5 34.1111111111111
};
\addlegendentry{Subm}
\addplot [semithick, color3, mark=*, mark size=1, mark options={solid}]
table {%
4.5 14.3333333333333
9.5 21.7777777777778
14.5 28.3333333333333
19.5 34
};
\addlegendentry{AG}
\addplot [semithick, color4, mark=*, mark size=1, mark options={solid}]
table {%
5 15
10 19.4444444444444
15 26.5555555555556
20 31.8888888888889
};
\addlegendentry{Quality}
\addplot [semithick, color5, mark=*, mark size=1, mark options={solid}]
table {%
5 3.33333333333333
10 5.11111111111111
15 7.55555555555556
20 9.88888888888889
};
\addlegendentry{Random}
\legend{}\end{axis}

\end{tikzpicture}
    }
    \hskip -2ex
    \subcaptionbox{Movies (unit costs)}{
    % This file was created by tikzplotlib v0.9.8.
\begin{tikzpicture}

\definecolor{color0}{rgb}{0.12156862745098,0.466666666666667,0.705882352941177}
\definecolor{color1}{rgb}{1,0.498039215686275,0.0549019607843137}
\definecolor{color2}{rgb}{0.580392156862745,0.403921568627451,0.741176470588235}
\definecolor{color3}{rgb}{0.83921568627451,0.152941176470588,0.156862745098039}
\definecolor{color4}{rgb}{0.172549019607843,0.627450980392157,0.172549019607843}
\definecolor{color5}{rgb}{0.549019607843137,0.337254901960784,0.294117647058824}

\begin{axis}[width = 0.3\textwidth, height = 0.25\textwidth, legend image post style={scale=0.4}, ylabel style={align=center},
legend cell align={left},
legend style={
  fill opacity=0.8,
  draw opacity=1,
  text opacity=1,
  at={(0.5,0.09)},
  anchor=south,
  draw=white!80!black
},
tick align=outside,
tick pos=left,
x grid style={white!69.0196078431373!black},
xlabel={Maximum budget},
xmin=3.7, xmax=21.3,
xtick style={color=black},
y grid style={white!69.0196078431373!black},
ylabel={\#activated\\users},
ymin=-2.18836460008123, ymax=102.075837602107,
ytick style={color=black}
, legend style={at={(0.25,2.1)},anchor=north,legend columns=2,font=\scriptsize}]
\path [draw=color0, semithick]
(axis cs:5,75.6054754146613)
--(axis cs:5,81.7278579186721);

\path [draw=color0, semithick]
(axis cs:10,84.3125300496322)
--(axis cs:10,89.4652477281455);

\path [draw=color0, semithick]
(axis cs:15,89.9386618765194)
--(axis cs:15,93.3946714568139);

\path [draw=color0, semithick]
(axis cs:20,90.4008223280092)
--(axis cs:20,95.3769554497685);

\path [draw=color1, semithick]
(axis cs:5,77.0840672577121)
--(axis cs:5,82.9159327422879);

\path [draw=color1, semithick]
(axis cs:10,86.7284527398496)
--(axis cs:10,91.4937694823726);

\path [draw=color1, semithick]
(axis cs:15,91.9535720416597)
--(axis cs:15,95.824205736118);

\path [draw=color1, semithick]
(axis cs:20,92.8856665383969)
--(axis cs:20,97.3365556838253);

\path [draw=color2, semithick]
(axis cs:5.5,73.8770885595585)
--(axis cs:5.5,79.6784669959971);

\path [draw=color2, semithick]
(axis cs:10.5,82.7426072935936)
--(axis cs:10.5,89.4796149286287);

\path [draw=color2, semithick]
(axis cs:15.5,90.4095588520429)
--(axis cs:15.5,92.479330036846);

\path [draw=color2, semithick]
(axis cs:20.5,89.9159011202714)
--(axis cs:20.5,95.417432213062);

\path [draw=color3, semithick]
(axis cs:4.5,73.8770885595585)
--(axis cs:4.5,79.6784669959971);

\path [draw=color3, semithick]
(axis cs:9.5,82.7426072935936)
--(axis cs:9.5,89.4796149286287);

\path [draw=color3, semithick]
(axis cs:14.5,90.4095588520429)
--(axis cs:14.5,92.479330036846);

\path [draw=color3, semithick]
(axis cs:19.5,89.9159011202714)
--(axis cs:19.5,95.417432213062);

\path [draw=color4, semithick]
(axis cs:5,68.8684785158208)
--(axis cs:5,72.6870770397348);

\path [draw=color4, semithick]
(axis cs:10,75.6182168965173)
--(axis cs:10,81.715116436816);

\path [draw=color4, semithick]
(axis cs:15,82.1297947724157)
--(axis cs:15,84.3146496720287);

\path [draw=color4, semithick]
(axis cs:20,84.5802115149725)
--(axis cs:20,88.3086773739164);

\path [draw=color5, semithick]
(axis cs:5,2.55091731820004)
--(axis cs:5,15.4490826818);

\path [draw=color5, semithick]
(axis cs:10,7.82899313695818)
--(axis cs:10,17.0598957519307);

\path [draw=color5, semithick]
(axis cs:15,13.2737640946256)
--(axis cs:15,21.1706803498188);

\path [draw=color5, semithick]
(axis cs:20,16.0480585021531)
--(axis cs:20,29.7297192756247);

\addplot [semithick, color0, mark=*, mark size=1, mark options={solid}]
table {%
5 78.6666666666667
10 86.8888888888889
15 91.6666666666667
20 92.8888888888889
};
\addlegendentry{Greedy-U}
\addplot [semithick, color1, mark=*, mark size=1, mark options={solid}]
table {%
5 80
10 89.1111111111111
15 93.8888888888889
20 95.1111111111111
};
\addlegendentry{Greedy-W}
\addplot [semithick, color2, mark=*, mark size=1, mark options={solid}]
table {%
5.5 76.7777777777778
10.5 86.1111111111111
15.5 91.4444444444444
20.5 92.6666666666667
};
\addlegendentry{Subm}
\addplot [semithick, color3, mark=*, mark size=1, mark options={solid}]
table {%
4.5 76.7777777777778
9.5 86.1111111111111
14.5 91.4444444444444
19.5 92.6666666666667
};
\addlegendentry{AG}
\addplot [semithick, color4, mark=*, mark size=1, mark options={solid}]
table {%
5 70.7777777777778
10 78.6666666666667
15 83.2222222222222
20 86.4444444444444
};
\addlegendentry{Quality}
\addplot [semithick, color5, mark=*, mark size=1, mark options={solid}]
table {%
5 9
10 12.4444444444444
15 17.2222222222222
20 22.8888888888889
};
\addlegendentry{Random}
\legend{}\end{axis}

\end{tikzpicture}
    }
    \caption{\label{fig:music}
    Results of using the \mar problem formulation 
    for making a playlist of items.
    The goal is to maximize the number of activated users.
    \revise{The universe \V includes songs, movies or books.}
    A user (\revise{a 0--1 activation function $\fsm_i$}) is activated if they like at least one item among all items they consume within their budget.
    Markers are jittered horizontally to avoid overlap.
    }
\end{figure}
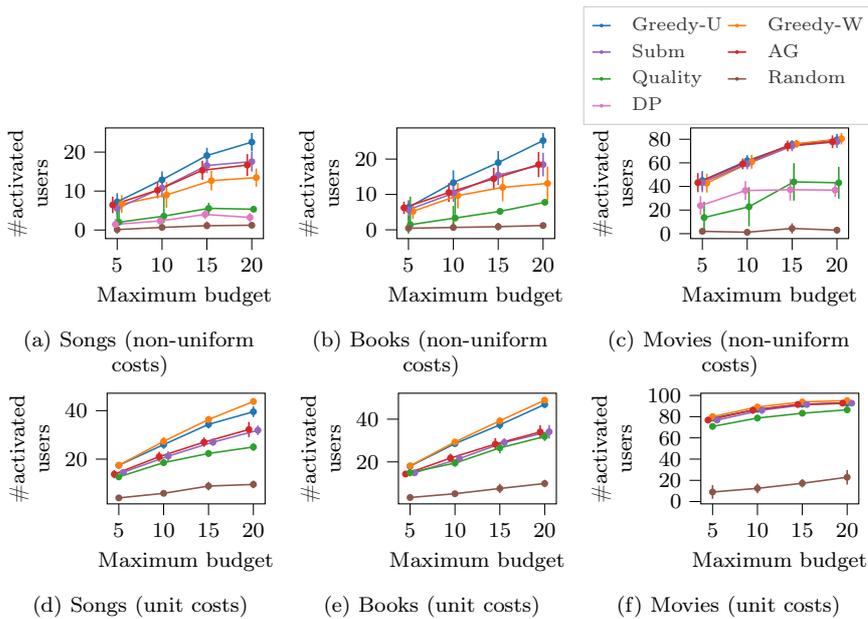

\subsection{Experiments with the max-submodular ranking (\msr) problem}

\para{Multiple intents re-ranking.}
We simulate a web-page ranking application for documents in the 20 Newsgroups dataset~\citep{UCI}.
For each newsgroup, we treat its title as a query, and collect documents that contains the query.
We extract 5 topics from the collected documents by means of LDA model \citep{blei2003latent}.
Subsequently, each topic (i.e., its top 20 keywords) is considered as a potential user intent, and the submodular utility for a particular topic when given a set of documents is the coverage rate of its top keywords.
We aim to find a ranking of documents that maximize the total utility of all user intents.
As in the previous experiment, we generate a random budget for each user intent, 
i.e., the maximum number of documents the potential user will read, from 1 to a given maximum budget.
For an additional non-uniform cost scenario,
we use the document length as the cost for reading a document, and accordingly multiply the budget by the average document length.

The results of our experiment are shown in Figure \ref{fig:news}, 
where we report the average performance across all newsgroups.
In the unit-cost scenario, the top-contender algorithms have close performance.
This is due to the overwhelming advantage of lengthy documents that contain more words and produce higher utility.
In the more realistic non-uniform cost scenario, our algorithms, \greedy and \wgreedy, achieve the best performance.
\quality algorithm behaves the worst as it fails to consider the cost of items, and its first-rank lengthy document exceeds the user budget most of the time.

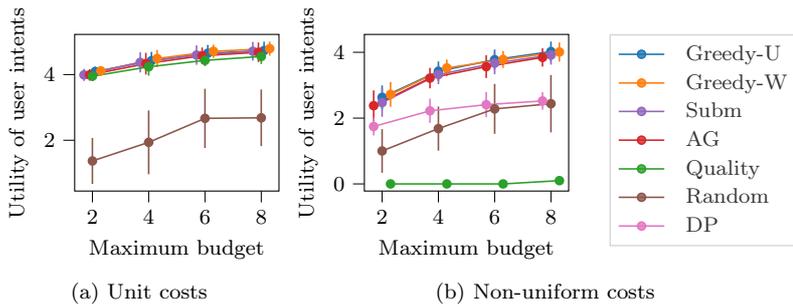
\begin{figure}[t]
    \centering
    \subcaptionbox{Unit costs}{
    % This file was created by tikzplotlib v0.9.8.
\begin{tikzpicture}

\definecolor{color0}{rgb}{0.12156862745098,0.466666666666667,0.705882352941177}
\definecolor{color1}{rgb}{1,0.498039215686275,0.0549019607843137}
\definecolor{color2}{rgb}{0.580392156862745,0.403921568627451,0.741176470588235}
\definecolor{color3}{rgb}{0.83921568627451,0.152941176470588,0.156862745098039}
\definecolor{color4}{rgb}{0.172549019607843,0.627450980392157,0.172549019607843}
\definecolor{color5}{rgb}{0.549019607843137,0.337254901960784,0.294117647058824}

\begin{axis}[width = 0.35\textwidth, height = 0.3\textwidth, 
legend cell align={left},
legend style={
  fill opacity=0.8,
  draw opacity=1,
  text opacity=1,
  at={(0.5,0.09)},
  anchor=south,
  draw=white!80!black
},
tick align=outside,
tick pos=left,
x grid style={white!69.0196078431373!black},
xlabel={Maximum budget},
xmin=1.37, xmax=8.63,
xtick style={color=black},
y grid style={white!69.0196078431373!black},
ylabel={Utility of user intents},
ymin=0.453452752226354, ymax=5.22300465991638,
ytick style={color=black}
, legend style={at={(1.2,1)}, anchor=north west}]
\path [draw=color0, semithick]
(axis cs:2.1,3.94221370897776)
--(axis cs:2.1,4.24111962435558);

\path [draw=color0, semithick]
(axis cs:4.1,4.15210096615027)
--(axis cs:4.1,4.69123236718306);

\path [draw=color0, semithick]
(axis cs:6.1,4.39697129630376)
--(axis cs:6.1,4.90302870369624);

\path [draw=color0, semithick]
(axis cs:8.1,4.46113126877804)
--(axis cs:8.1,5.00220206455529);

\path [draw=color1, semithick]
(axis cs:2.3,3.97247674429266)
--(axis cs:2.3,4.25085658904067);

\path [draw=color1, semithick]
(axis cs:4.3,4.22548163137764)
--(axis cs:4.3,4.75118503528902);

\path [draw=color1, semithick]
(axis cs:6.3,4.50628568415505)
--(axis cs:6.3,4.91038098251162);

\path [draw=color1, semithick]
(axis cs:8.3,4.57379315406953)
--(axis cs:8.3,5.00620684593047);

\path [draw=color2, semithick]
(axis cs:1.7,3.80613515853666)
--(axis cs:1.7,4.17386484146334);

\path [draw=color2, semithick]
(axis cs:3.7,4.06469037596684)
--(axis cs:3.7,4.67864295736649);

\path [draw=color2, semithick]
(axis cs:5.7,4.29928077984582)
--(axis cs:5.7,4.89071922015418);

\path [draw=color2, semithick]
(axis cs:7.7,4.42400678955333)
--(axis cs:7.7,4.99599321044667);

\path [draw=color3, semithick]
(axis cs:1.9,3.81431829759686)
--(axis cs:1.9,4.16568170240314);

\path [draw=color3, semithick]
(axis cs:3.9,3.99954730815092)
--(axis cs:3.9,4.66378602518242);

\path [draw=color3, semithick]
(axis cs:5.9,4.32278714068096)
--(axis cs:5.9,4.8238795259857);

\path [draw=color3, semithick]
(axis cs:7.9,4.38019901473753)
--(axis cs:7.9,4.97646765192914);

\path [draw=color4, semithick]
(axis cs:2,3.80333911857247)
--(axis cs:2,4.09666088142754);

\path [draw=color4, semithick]
(axis cs:4,3.9668951232645)
--(axis cs:4,4.50643821006883);

\path [draw=color4, semithick]
(axis cs:6,4.25759566949948)
--(axis cs:6,4.60573766383385);

\path [draw=color4, semithick]
(axis cs:8,4.3208167741974)
--(axis cs:8,4.78251655913593);

\path [draw=color5, semithick]
(axis cs:2,0.670250566212264)
--(axis cs:2,2.06974943378774);

\path [draw=color5, semithick]
(axis cs:4,0.965703751406389)
--(axis cs:4,2.90762958192694);

\path [draw=color5, semithick]
(axis cs:6,1.76566382281228)
--(axis cs:6,3.56766951052106);

\path [draw=color5, semithick]
(axis cs:8,1.82169348482544)
--(axis cs:8,3.54497318184122);

\addplot [semithick, color0, mark=*, mark size=1.5, mark options={solid}]
table {%
2.1 4.09166666666667
4.1 4.42166666666667
6.1 4.65
8.1 4.73166666666667
};
\addlegendentry{Greedy-U}
\addplot [semithick, color1, mark=*, mark size=1.5, mark options={solid}]
table {%
2.3 4.11166666666667
4.3 4.48833333333333
6.3 4.70833333333333
8.3 4.79
};
\addlegendentry{Greedy-W}
\addplot [semithick, color2, mark=*, mark size=1.5, mark options={solid}]
table {%
1.7 3.99
3.7 4.37166666666667
5.7 4.595
7.7 4.71
};
\addlegendentry{Subm}
\addplot [semithick, color3, mark=*, mark size=1.5, mark options={solid}]
table {%
1.9 3.99
3.9 4.33166666666667
5.9 4.57333333333333
7.9 4.67833333333333
};
\addlegendentry{AG}
\addplot [semithick, color4, mark=*, mark size=1.5, mark options={solid}]
table {%
2 3.95
4 4.23666666666667
6 4.43166666666667
8 4.55166666666667
};
\addlegendentry{Quality}
\addplot [semithick, color5, mark=*, mark size=1.5, mark options={solid}]
table {%
2 1.37
4 1.93666666666667
6 2.66666666666667
8 2.68333333333333
};
\addlegendentry{Random}
\legend{}\end{axis}

\end{tikzpicture}
    }
    \hskip -2ex
    \subcaptionbox{Non-uniform costs}{
    % This file was created by tikzplotlib v0.9.8.
\begin{tikzpicture}

\definecolor{color0}{rgb}{0.12156862745098,0.466666666666667,0.705882352941177}
\definecolor{color1}{rgb}{1,0.498039215686275,0.0549019607843137}
\definecolor{color2}{rgb}{0.580392156862745,0.403921568627451,0.741176470588235}
\definecolor{color3}{rgb}{0.83921568627451,0.152941176470588,0.156862745098039}
\definecolor{color4}{rgb}{0.172549019607843,0.627450980392157,0.172549019607843}
\definecolor{color5}{rgb}{0.549019607843137,0.337254901960784,0.294117647058824}
\definecolor{color6}{rgb}{0.890196078431372,0.466666666666667,0.76078431372549}

\begin{axis}[width = 0.35\textwidth, height = 0.3\textwidth, 
legend cell align={left},
legend style={
  fill opacity=0.8,
  draw opacity=1,
  text opacity=1,
  at={(0.97,0.03)},
  anchor=south east,
  draw=white!80!black
},
tick align=outside,
tick pos=left,
x grid style={white!69.0196078431373!black},
xlabel={Maximum budget},
xmin=1.37, xmax=8.63,
xtick style={color=black},
y grid style={white!69.0196078431373!black},
ylabel={Utility of user intents},
ymin=-0.216331974548822, ymax=4.54297146552525,
ytick style={color=black}
, legend style={at={(1.2,1)}, anchor=north west}]
\path [draw=color0, semithick]
(axis cs:2,2.26557881316854)
--(axis cs:2,2.99442118683146);

\path [draw=color0, semithick]
(axis cs:4,3.12438809144744)
--(axis cs:4,3.72227857521923);

\path [draw=color0, semithick]
(axis cs:6,3.46573310871404)
--(axis cs:6,4.08426689128596);

\path [draw=color0, semithick]
(axis cs:8,3.7166938423569)
--(axis cs:8,4.32663949097643);

\path [draw=color1, semithick]
(axis cs:2.3,2.34129808179564)
--(axis cs:2.3,3.09870191820436);

\path [draw=color1, semithick]
(axis cs:4.3,3.26885472997143)
--(axis cs:4.3,3.77781193669524);

\path [draw=color1, semithick]
(axis cs:6.3,3.52377253161344)
--(axis cs:6.3,4.03956080171989);

\path [draw=color1, semithick]
(axis cs:8.3,3.71949389142083)
--(axis cs:8.3,4.29717277524584);

\path [draw=color2, semithick]
(axis cs:2,2.03948571864621)
--(axis cs:2,2.90051428135379);

\path [draw=color2, semithick]
(axis cs:4,3.0309328111346)
--(axis cs:4,3.61573385553206);

\path [draw=color2, semithick]
(axis cs:6,3.33089245616948)
--(axis cs:6,4.00910754383052);

\path [draw=color2, semithick]
(axis cs:8,3.62516164431438)
--(axis cs:8,4.20483835568562);

\path [draw=color3, semithick]
(axis cs:1.7,1.91151070870979)
--(axis cs:1.7,2.84182262462354);

\path [draw=color3, semithick]
(axis cs:3.7,2.90373469436023)
--(axis cs:3.7,3.52626530563977);

\path [draw=color3, semithick]
(axis cs:5.7,3.21751712392954)
--(axis cs:5.7,3.90914954273713);

\path [draw=color3, semithick]
(axis cs:7.7,3.56576588481571)
--(axis cs:7.7,4.12090078185096);

\path [draw=color4, semithick]
(axis cs:2.3,0)
--(axis cs:2.3,0);

\path [draw=color4, semithick]
(axis cs:4.3,0)
--(axis cs:4.3,0);

\path [draw=color4, semithick]
(axis cs:6.3,0)
--(axis cs:6.3,0);

\path [draw=color4, semithick]
(axis cs:8.3,0.0563451057983247)
--(axis cs:8.3,0.143654894201675);

\path [draw=color5, semithick]
(axis cs:2,0.338649443530337)
--(axis cs:2,1.664683889803);

\path [draw=color5, semithick]
(axis cs:4,1.01045401539494)
--(axis cs:4,2.3528793179384);

\path [draw=color5, semithick]
(axis cs:6,1.52118612066079)
--(axis cs:6,3.03881387933921);

\path [draw=color5, semithick]
(axis cs:8,1.56536486963445)
--(axis cs:8,3.30796846369888);

\path [draw=color6, semithick]
(axis cs:1.7,1.47301969227177)
--(axis cs:1.7,2.01698030772823);

\path [draw=color6, semithick]
(axis cs:3.7,1.853877154377)
--(axis cs:3.7,2.58945617895633);

\path [draw=color6, semithick]
(axis cs:5.7,2.02820075794167)
--(axis cs:5.7,2.788465908725);

\path [draw=color6, semithick]
(axis cs:7.7,2.25618514019693)
--(axis cs:7.7,2.78714819313641);

\addplot [semithick, color0, mark=*, mark size=1.5, mark options={solid}]
table {%
2 2.63
4 3.42333333333333
6 3.775
8 4.02166666666667
};
\addlegendentry{Greedy-U}
\addplot [semithick, color1, mark=*, mark size=1.5, mark options={solid}]
table {%
2.3 2.72
4.3 3.52333333333333
6.3 3.78166666666667
8.3 4.00833333333333
};
\addlegendentry{Greedy-W}
\addplot [semithick, color2, mark=*, mark size=1.5, mark options={solid}]
table {%
2 2.47
4 3.32333333333333
6 3.67
8 3.915
};
\addlegendentry{Subm}
\addplot [semithick, color3, mark=*, mark size=1.5, mark options={solid}]
table {%
1.7 2.37666666666667
3.7 3.215
5.7 3.56333333333333
7.7 3.84333333333333
};
\addlegendentry{AG}
\addplot [semithick, color4, mark=*, mark size=1.5, mark options={solid}]
table {%
2.3 1.11022302462516e-16
4.3 1.11022302462516e-16
6.3 1.11022302462516e-16
8.3 0.1
};
\addlegendentry{Quality}
\addplot [semithick, color5, mark=*, mark size=1.5, mark options={solid}]
table {%
2 1.00166666666667
4 1.68166666666667
6 2.28
8 2.43666666666667
};
\addlegendentry{Random}
\addplot [semithick, color6, mark=*, mark size=1.5, mark options={solid}]
table {%
1.7 1.745
3.7 2.22166666666667
5.7 2.40833333333333
7.7 2.52166666666667
};
\addlegendentry{DP}
%\legend{}
\end{axis}

\end{tikzpicture}
    }
    \caption{\label{fig:news}\msr for multiple intents re-ranking in web page ranking.
    The goal is to maximize the total utility of all user intents within their individual reading budget.
    \revise{The universe \V includes documents}.
    The utility of a user intent (\revise{a coverage function $\fsm_i$}) is represented by the coverage rate of its top keywords.
    Markers are jittered horizontally to avoid overlap.
    }
\end{figure}

\smallskip
\para{Sequential active learning.}
Active learning seeks to make label queries on only a small number of informative data points in order 
to maximize  model performance.
In particular, for the $k$-nearest neighbors ($k$NN) model,
an intuitive measure for informativeness of a set of labeled data points
is the average distance from an unlabeled data point to its closest labeled point, 
i.e., the facility-location function \citep{wei2015submodularity}.
We refer to this average distance as the \emph{radius}.
Thus, the active-learning task can be naturally formulated 
as labeling a small subset of data to maximize the radius reduction.
Note that the reduction of the radius by labeling a subset of data points is clearly non-decreasing and submodular.

In our setting, we assume that we have access to multiple models that are trained on the same labeled data, 
and we aim to label data sequentially to maximize the total reduction in the radii among all models.
This happens, for example, when each model runs on a different subset of features.
Interestingly, in this case each model can be seen as a student with different learning capacity, and a teacher tries to optimize the classroom teaching by feeding them labeled data \citep{zhu2017no}.
We evaluate the performance of active-learning $k$NNs ($k=1$) 
with Euclidean distance in the Handwritten Digits dataset~\citep{UCI}.
Each $k$NN model adopts a different strategy in unsupervised feature selection, 
such as variance thresholding, PCA, and feature agglomeration.
Again, we generate a random query budget for each model and a random cost (from 1 to 10) for labeling each data point.

As we can see in Figure \ref{fig:nist}, all greedy algorithms are very effective in reducing the radii.
The correlation between the radius reduction and model accuracy (over testing data) is obvious.
Note that the \random algorithm is a standard strong baseline in data subset selection, 
which is outperformed by the greedy algorithms by a large margin.
The comparison becomes more evident in the non-uniform cost scenario, as the \random algorithm fails to take into account the item costs.

\begin{figure}[t]
    \centering
    \subcaptionbox{Unit costs}{
    % This file was created by tikzplotlib v0.9.8.
\begin{tikzpicture}

\definecolor{color0}{rgb}{0.12156862745098,0.466666666666667,0.705882352941177}
\definecolor{color1}{rgb}{1,0.498039215686275,0.0549019607843137}
\definecolor{color2}{rgb}{0.580392156862745,0.403921568627451,0.741176470588235}
\definecolor{color3}{rgb}{0.83921568627451,0.152941176470588,0.156862745098039}
\definecolor{color4}{rgb}{0.172549019607843,0.627450980392157,0.172549019607843}
\definecolor{color5}{rgb}{0.549019607843137,0.337254901960784,0.294117647058824}

\begin{groupplot}[width = 0.35\textwidth, height = 0.3\textwidth, ylabel style={align=center}, group/vertical sep=20pt, xtick={25,50,75,100},group style={group size=1 by 2}]
\nextgroupplot[
scaled x ticks=manual:{}{\pgfmathparse{#1}},
tick align=outside,
tick pos=left,
x grid style={white!69.0196078431373!black},
xmin=16.85, xmax=108.15,
xtick style={color=black},
xticklabels={},
y grid style={white!69.0196078431373!black},
ylabel={Reduction of radii},
ymin=1.32407024988439, ymax=2.40464203270156,
ytick style={color=black}
, legend style={at={(1.2,1)}, anchor=north west}]
\path [draw=color0, semithick]
(axis cs:25,1.61195320770008)
--(axis cs:25,2.01190648861985);

\path [draw=color0, semithick]
(axis cs:50,1.94025551124264)
--(axis cs:50,2.19585362461796);

\path [draw=color0, semithick]
(axis cs:75,2.19850030442302)
--(axis cs:75,2.25519208187813);

\path [draw=color0, semithick]
(axis cs:100,2.22440122378592)
--(axis cs:100,2.3555251334826);

\path [draw=color1, semithick]
(axis cs:29,1.62149077292936)
--(axis cs:29,2.01905514426926);

\path [draw=color1, semithick]
(axis cs:54,1.94428806814415)
--(axis cs:54,2.20074774058681);

\path [draw=color1, semithick]
(axis cs:79,2.20259689143582)
--(axis cs:79,2.25451352624224);

\path [draw=color1, semithick]
(axis cs:104,2.22814040582531)
--(axis cs:104,2.35487048404817);

\path [draw=color2, semithick]
(axis cs:23,1.60940893308703)
--(axis cs:23,2.00783657637483);

\path [draw=color2, semithick]
(axis cs:48,1.93612393513076)
--(axis cs:48,2.19413390180669);

\path [draw=color2, semithick]
(axis cs:73,2.19635680010895)
--(axis cs:73,2.25218902895098);

\path [draw=color2, semithick]
(axis cs:98,2.22034607959515)
--(axis cs:98,2.35487243296956);

\path [draw=color3, semithick]
(axis cs:21,1.60939404122392)
--(axis cs:21,2.00791821750606);

\path [draw=color3, semithick]
(axis cs:46,1.93642719911866)
--(axis cs:46,2.19442252608159);

\path [draw=color3, semithick]
(axis cs:71,2.19672809197939)
--(axis cs:71,2.25268936363214);

\path [draw=color3, semithick]
(axis cs:96,2.22052860662522)
--(axis cs:96,2.35507210490983);

\path [draw=color4, semithick]
(axis cs:27,1.50958966972207)
--(axis cs:27,1.72533868224139);

\path [draw=color4, semithick]
(axis cs:52,1.71489711255072)
--(axis cs:52,1.85539001962504);

\path [draw=color4, semithick]
(axis cs:77,1.87572489339678)
--(axis cs:77,1.94514027349563);

\path [draw=color4, semithick]
(axis cs:102,1.93165489535032)
--(axis cs:102,2.06550211709546);

\path [draw=color5, semithick]
(axis cs:25,1.37318714910335)
--(axis cs:25,1.83289475917498);

\path [draw=color5, semithick]
(axis cs:50,1.74074420395628)
--(axis cs:50,2.03006468310985);

\path [draw=color5, semithick]
(axis cs:75,2.06611794542529)
--(axis cs:75,2.12979771383989);

\path [draw=color5, semithick]
(axis cs:100,2.08345267371251)
--(axis cs:100,2.25384519739206);

\addplot [semithick, color0, mark=*, mark size=1.5, mark options={solid}]
table {%
25 1.81192984815996
50 2.0680545679303
75 2.22684619315057
100 2.28996317863426
};
\addplot [semithick, color1, mark=*, mark size=1.5, mark options={solid}]
table {%
29 1.82027295859931
54 2.07251790436548
79 2.22855520883903
104 2.29150544493674
};
\addplot [semithick, color2, mark=*, mark size=1.5, mark options={solid}]
table {%
23 1.80862275473093
48 2.06512891846872
73 2.22427291452997
98 2.28760925628236
};
\addplot [semithick, color3, mark=*, mark size=1.5, mark options={solid}]
table {%
21 1.80865612936499
46 2.06542486260013
71 2.22470872780577
96 2.28780035576753
};
\addplot [semithick, color4, mark=*, mark size=1.5, mark options={solid}]
table {%
27 1.61746417598173
52 1.78514356608788
77 1.91043258344621
102 1.99857850622289
};
\addplot [semithick, color5, mark=*, mark size=1.5, mark options={solid}]
table {%
25 1.60304095413917
50 1.88540444353307
75 2.09795782963259
100 2.16864893555229
};

\nextgroupplot[
legend cell align={left},
legend style={
  fill opacity=0.8,
  draw opacity=1,
  text opacity=1,
  at={(0.5,0.09)},
  anchor=south,
  draw=white!80!black
},
tick align=outside,
tick pos=left,
x grid style={white!69.0196078431373!black},
xlabel={Maximum budget},
xmin=16.85, xmax=108.15,
xtick style={color=black},
y grid style={white!69.0196078431373!black},
ylabel={kNN accuracy},
ymin=0, ymax=1,
ytick style={color=black}
, legend style={at={(1.2,1)}, anchor=north west}]
\path [draw=color0, semithick]
(axis cs:25,0.252084661107206)
--(axis cs:25,0.71384126481872);

\path [draw=color0, semithick]
(axis cs:50,0.608852718964727)
--(axis cs:50,0.877896252228689);

\path [draw=color0, semithick]
(axis cs:75,0.853026335240518)
--(axis cs:75,0.872570372578412);

\path [draw=color0, semithick]
(axis cs:100,0.849558850735487)
--(axis cs:100,0.929206581363278);

\path [draw=color1, semithick]
(axis cs:29,0.240284310097226)
--(axis cs:29,0.71839881747479);

\path [draw=color1, semithick]
(axis cs:54,0.608095914244695)
--(axis cs:54,0.87717157546724);

\path [draw=color1, semithick]
(axis cs:79,0.851201436076259)
--(axis cs:79,0.877029016598638);

\path [draw=color1, semithick]
(axis cs:104,0.850078057491089)
--(axis cs:104,0.928687374607677);

\path [draw=color2, semithick]
(axis cs:23,0.256618965494166)
--(axis cs:23,0.719841939855629);

\path [draw=color2, semithick]
(axis cs:48,0.600974147926307)
--(axis cs:48,0.880178115448179);

\path [draw=color2, semithick]
(axis cs:73,0.853921285361849)
--(axis cs:73,0.878095175543501);

\path [draw=color2, semithick]
(axis cs:98,0.857874102321983)
--(axis cs:98,0.927146473809704);

\path [draw=color3, semithick]
(axis cs:21,0.256641598214859)
--(axis cs:21,0.719654698081438);

\path [draw=color3, semithick]
(axis cs:46,0.601058445919318)
--(axis cs:46,0.880093817455168);

\path [draw=color3, semithick]
(axis cs:71,0.855351108657514)
--(axis cs:71,0.879463706157301);

\path [draw=color3, semithick]
(axis cs:96,0.858407244142891)
--(axis cs:96,0.928094813470278);

\path [draw=color4, semithick]
(axis cs:27,0.116759813981539)
--(axis cs:27,0.171635247746856);

\path [draw=color4, semithick]
(axis cs:52,0.192198331407631)
--(axis cs:52,0.269694672707596);

\path [draw=color4, semithick]
(axis cs:77,0.308244751819177)
--(axis cs:77,0.379162655588231);

\path [draw=color4, semithick]
(axis cs:102,0.373296835244823)
--(axis cs:102,0.525468596853943);

\path [draw=color5, semithick]
(axis cs:25,0.186607958581717)
--(axis cs:25,0.486972288331864);

\path [draw=color5, semithick]
(axis cs:50,0.464855186117628)
--(axis cs:50,0.636873208944101);

\path [draw=color5, semithick]
(axis cs:75,0.6947514057887)
--(axis cs:75,0.784425548943811);

\path [draw=color5, semithick]
(axis cs:100,0.703856804559072)
--(axis cs:100,0.855813977333932);

\addplot [semithick, color0, dashed, mark=*, mark size=1.5, mark options={solid}]
table {%
25 0.482962962962963
50 0.743374485596708
75 0.862798353909465
100 0.889382716049383
};
\addlegendentry{Greedy-U}
\addplot [semithick, color1, dashed, mark=*, mark size=1.5, mark options={solid}]
table {%
29 0.479341563786008
54 0.742633744855967
79 0.864115226337449
104 0.889382716049383
};
\addlegendentry{Greedy-W}
\addplot [semithick, color2, dashed, mark=*, mark size=1.5, mark options={solid}]
table {%
23 0.488230452674897
48 0.740576131687243
73 0.866008230452675
98 0.892510288065844
};
\addlegendentry{Subm}
\addplot [semithick, color3, dashed, mark=*, mark size=1.5, mark options={solid}]
table {%
21 0.488148148148148
46 0.740576131687243
71 0.867407407407407
96 0.893251028806584
};
\addlegendentry{AG}
\addplot [semithick, color4, dashed, mark=*, mark size=1.5, mark options={solid}]
table {%
27 0.144197530864197
52 0.230946502057613
77 0.343703703703704
102 0.449382716049383
};
\addlegendentry{Quality}
\addplot [semithick, color5, dashed, mark=*, mark size=1.5, mark options={solid}]
table {%
25 0.33679012345679
50 0.550864197530864
75 0.739588477366255
100 0.779835390946502
};
\addlegendentry{Random}
\legend{}\end{groupplot}

\end{tikzpicture}
    }
    \hskip -2ex
    \subcaptionbox{Non-uniform costs}{
    % This file was created by tikzplotlib v0.9.8.
\begin{tikzpicture}

\definecolor{color0}{rgb}{0.12156862745098,0.466666666666667,0.705882352941177}
\definecolor{color1}{rgb}{1,0.498039215686275,0.0549019607843137}
\definecolor{color2}{rgb}{0.580392156862745,0.403921568627451,0.741176470588235}
\definecolor{color3}{rgb}{0.83921568627451,0.152941176470588,0.156862745098039}
\definecolor{color4}{rgb}{0.172549019607843,0.627450980392157,0.172549019607843}
\definecolor{color5}{rgb}{0.549019607843137,0.337254901960784,0.294117647058824}
\definecolor{color6}{rgb}{0.890196078431372,0.466666666666667,0.76078431372549}

\begin{groupplot}[width = 0.35\textwidth, height = 0.3\textwidth, ylabel style={align=center}, group/vertical sep=20pt, xtick={25,50,75,100},group style={group size=1 by 2}]
\nextgroupplot[
scaled x ticks=manual:{}{\pgfmathparse{#1}},
tick align=outside,
tick pos=left,
x grid style={white!69.0196078431373!black},
xmin=16.85, xmax=108.15,
xtick style={color=black},
xticklabels={},
y grid style={white!69.0196078431373!black},
ylabel={Reduction of radii},
ymin=-0.0642742223580677, ymax=2.5456143343521,
ytick style={color=black}
, legend style={at={(1.2,1)}, anchor=north west}]
\path [draw=color0, semithick]
(axis cs:25,0.986686347678879)
--(axis cs:25,1.89523891121217);

\path [draw=color0, semithick]
(axis cs:50,1.48528905337995)
--(axis cs:50,2.25546468234064);

\path [draw=color0, semithick]
(axis cs:75,2.15550627553066)
--(axis cs:75,2.19882781501917);

\path [draw=color0, semithick]
(axis cs:100,2.17383240559985)
--(axis cs:100,2.29443451431329);

\path [draw=color1, semithick]
(axis cs:29,0.981851894969747)
--(axis cs:29,1.90683573714328);

\path [draw=color1, semithick]
(axis cs:54,1.48549298275303)
--(axis cs:54,2.2576425803877);

\path [draw=color1, semithick]
(axis cs:79,2.15704280656911)
--(axis cs:79,2.19922540146853);

\path [draw=color1, semithick]
(axis cs:104,2.17760641315988)
--(axis cs:104,2.29389449472615);

\path [draw=color2, semithick]
(axis cs:23,0.982764039847151)
--(axis cs:23,1.89336849442032);

\path [draw=color2, semithick]
(axis cs:48,1.48215761039224)
--(axis cs:48,2.25298027536779);

\path [draw=color2, semithick]
(axis cs:73,2.15431376599168)
--(axis cs:73,2.1968954044525);

\path [draw=color2, semithick]
(axis cs:98,2.17174671153447)
--(axis cs:98,2.29332883647257);

\path [draw=color3, semithick]
(axis cs:21,0.982742364966567)
--(axis cs:21,1.89345613689405);

\path [draw=color3, semithick]
(axis cs:46,1.48215807417517)
--(axis cs:46,2.25298698305861);

\path [draw=color3, semithick]
(axis cs:71,2.15480844120455)
--(axis cs:71,2.19684920831328);

\path [draw=color3, semithick]
(axis cs:96,2.17173866858647)
--(axis cs:96,2.29316725419552);

\path [draw=color4, semithick]
(axis cs:27,0.0863881821320435)
--(axis cs:27,1.366769800454);

\path [draw=color4, semithick]
(axis cs:52,1.13991807052537)
--(axis cs:52,1.6741560803796);

\path [draw=color4, semithick]
(axis cs:77,1.6445049837483)
--(axis cs:77,1.68957261868894);

\path [draw=color4, semithick]
(axis cs:102,1.61157992868005)
--(axis cs:102,1.7775047681293);

\path [draw=color5, semithick]
(axis cs:25,0.0543570756742128)
--(axis cs:25,1.50855242527407);

\path [draw=color5, semithick]
(axis cs:50,1.13141245439712)
--(axis cs:50,1.73898934457288);

\path [draw=color5, semithick]
(axis cs:75,1.66734258836641)
--(axis cs:75,1.71905992603685);

\path [draw=color5, semithick]
(axis cs:100,1.76206300627401)
--(axis cs:100,1.88228549446562);

\path [draw=color6, semithick]
(axis cs:23,0.719733251548797)
--(axis cs:23,1.44923022743473);

\path [draw=color6, semithick]
(axis cs:48,0.93937705932044)
--(axis cs:48,1.48189108895933);

\path [draw=color6, semithick]
(axis cs:73,1.44465935806445)
--(axis cs:73,2.28670840254454);

\path [draw=color6, semithick]
(axis cs:98,1.4125416224088)
--(axis cs:98,2.42698303631982);

\addplot [semithick, color0, mark=*, mark size=1.5, mark options={solid}]
table {%
25 1.44096262944552
50 1.87037686786029
75 2.17716704527492
100 2.23413345995657
};
\addplot [semithick, color1, mark=*, mark size=1.5, mark options={solid}]
table {%
29 1.44434381605652
54 1.87156778157037
79 2.17813410401882
104 2.23575045394302
};
\addplot [semithick, color2, mark=*, mark size=1.5, mark options={solid}]
table {%
23 1.43806626713374
48 1.86756894288001
73 2.17560458522209
98 2.23253777400352
};
\addplot [semithick, color3, mark=*, mark size=1.5, mark options={solid}]
table {%
21 1.43809925093031
46 1.86757252861689
71 2.17582882475892
96 2.232452961391
};
\addplot [semithick, color4, mark=*, mark size=1.5, mark options={solid}]
table {%
27 0.726578991293021
52 1.40703707545248
77 1.66703880121862
102 1.69454234840467
};
\addplot [semithick, color5, mark=*, mark size=1.5, mark options={solid}]
table {%
25 0.781454750474144
50 1.435200899485
75 1.69320125720163
100 1.82217425036982
};
\addplot [semithick, color6, mark=*, mark size=1.5, mark options={solid}]
table {%
23 1.08448173949176
48 1.21063407413989
73 1.8656838803045
98 1.91976232936431
};

\nextgroupplot[
legend cell align={left},
legend style={
  fill opacity=0.8,
  draw opacity=1,
  text opacity=1,
  at={(0.5,0.09)},
  anchor=south,
  draw=white!80!black
},
tick align=outside,
tick pos=left,
x grid style={white!69.0196078431373!black},
xlabel={Maximum budget},
xmin=16.85, xmax=108.15,
xtick style={color=black},
y grid style={white!69.0196078431373!black},
ylabel={kNN accuracy},
ymin=0, ymax=1,
ytick style={color=black}
, legend style={at={(1.2,1)}, anchor=north west}]
\path [draw=color0, semithick]
(axis cs:25,0.164715790308535)
--(axis cs:25,0.665078448374593);

\path [draw=color0, semithick]
(axis cs:50,0.549797428020415)
--(axis cs:50,0.827486522596869);

\path [draw=color0, semithick]
(axis cs:75,0.813333316908953)
--(axis cs:75,0.84164610695936);

\path [draw=color0, semithick]
(axis cs:100,0.815725105128049)
--(axis cs:100,0.886661726147671);

\path [draw=color1, semithick]
(axis cs:29,0.166414938612131)
--(axis cs:29,0.66634226303396);

\path [draw=color1, semithick]
(axis cs:54,0.548542736980389)
--(axis cs:54,0.819852324748006);

\path [draw=color1, semithick]
(axis cs:79,0.813130938421374)
--(axis cs:79,0.841848485446939);

\path [draw=color1, semithick]
(axis cs:104,0.811845778575304)
--(axis cs:104,0.888071916897946);

\path [draw=color2, semithick]
(axis cs:23,0.17230410224543)
--(axis cs:23,0.657654745491195);

\path [draw=color2, semithick]
(axis cs:48,0.555835058348302)
--(axis cs:48,0.825399509552933);

\path [draw=color2, semithick]
(axis cs:73,0.810346989357894)
--(axis cs:73,0.846607743152394);

\path [draw=color2, semithick]
(axis cs:98,0.814270189564371)
--(axis cs:98,0.889927341299827);

\path [draw=color3, semithick]
(axis cs:21,0.172322821430925)
--(axis cs:21,0.657471417252202);

\path [draw=color3, semithick]
(axis cs:46,0.555835058348302)
--(axis cs:46,0.825399509552933);

\path [draw=color3, semithick]
(axis cs:71,0.81013813003522)
--(axis cs:71,0.847310429635562);

\path [draw=color3, semithick]
(axis cs:96,0.814240944999955)
--(axis cs:96,0.889133540596752);

\path [draw=color4, semithick]
(axis cs:27,0.00601255463386214)
--(axis cs:27,0.0976911490698416);

\path [draw=color4, semithick]
(axis cs:52,0.0831113230437177)
--(axis cs:52,0.134008018520068);

\path [draw=color4, semithick]
(axis cs:77,0.138112344133538)
--(axis cs:77,0.159336215537244);

\path [draw=color4, semithick]
(axis cs:102,0.140933928043175)
--(axis cs:102,0.179724508170817);

\path [draw=color5, semithick]
(axis cs:25,0.005531900744257)
--(axis cs:25,0.188048346169323);

\path [draw=color5, semithick]
(axis cs:50,0.207535300866502)
--(axis cs:50,0.338143711479177);

\path [draw=color5, semithick]
(axis cs:75,0.306853977871516)
--(axis cs:75,0.419894993321899);

\path [draw=color5, semithick]
(axis cs:100,0.426237201793657)
--(axis cs:100,0.546602304379183);

\path [draw=color6, semithick]
(axis cs:23,0.0562076665035069)
--(axis cs:23,0.126179164772213);

\path [draw=color6, semithick]
(axis cs:48,0.0673827498094046)
--(axis cs:48,0.110395027968373);

\path [draw=color6, semithick]
(axis cs:73,0.160771075421825)
--(axis cs:73,0.910340035689286);

\path [draw=color6, semithick]
(axis cs:98,0.162555874512321)
--(axis cs:98,0.991188981454757);

\addplot [semithick, color0, dashed, mark=*, mark size=1.5, mark options={solid}]
table {%
25 0.414897119341564
50 0.688641975308642
75 0.827489711934156
100 0.85119341563786
};
\addlegendentry{Greedy-U}
\addplot [semithick, color1, dashed, mark=*, mark size=1.5, mark options={solid}]
table {%
29 0.416378600823045
54 0.684197530864197
79 0.827489711934156
104 0.849958847736625
};
\addlegendentry{Greedy-W}
\addplot [semithick, color2, dashed, mark=*, mark size=1.5, mark options={solid}]
table {%
23 0.414979423868313
48 0.690617283950617
73 0.828477366255144
98 0.852098765432099
};
\addlegendentry{Subm}
\addplot [semithick, color3, dashed, mark=*, mark size=1.5, mark options={solid}]
table {%
21 0.414897119341564
46 0.690617283950617
71 0.828724279835391
96 0.851687242798354
};
\addlegendentry{AG}
\addplot [semithick, color4, dashed, mark=*, mark size=1.5, mark options={solid}]
table {%
27 0.0518518518518518
52 0.108559670781893
77 0.148724279835391
102 0.160329218106996
};
\addlegendentry{Quality}
\addplot [semithick, color5, dashed, mark=*, mark size=1.5, mark options={solid}]
table {%
25 0.0967901234567901
50 0.272839506172839
75 0.363374485596708
100 0.48641975308642
};
\addlegendentry{Random}
\addplot [semithick, color6, dashed, mark=*, mark size=1.5, mark options={solid}]
table {%
23 0.09119341563786
48 0.0888888888888888
73 0.535555555555556
98 0.576872427983539
};
\addlegendentry{DP}
%\legend{}
\end{groupplot}

\end{tikzpicture}
    }
    \caption{\label{fig:nist}\msr for sequential data subset selection for $k$NN models.
    The goal is to boost the average predictive accuracy of $k$NN models.
    \revise{The universe \V includes all data points.}
    \revise{The sum of the surrogate objective function $\fsm_i$ (reduction of radii) for each model is optimized.}
    Markers are jittered horizontally to avoid overlap.
    }
\end{figure}
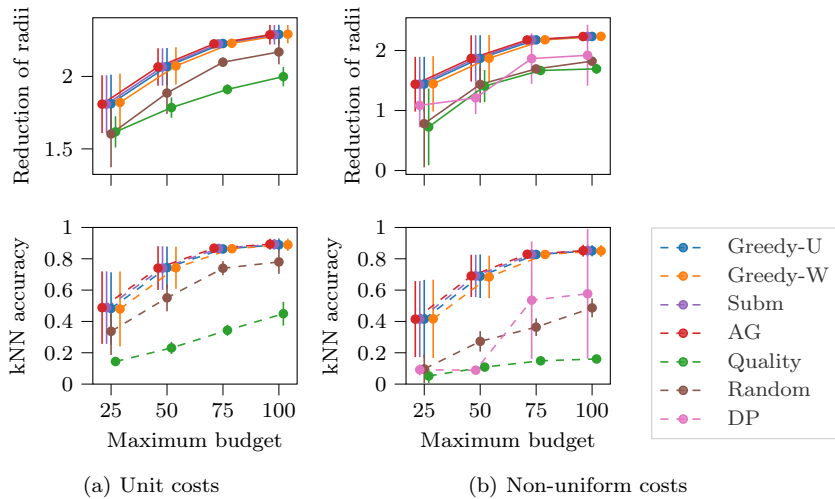

\subsection{Running time}
\label{subsection:runtime}

We examine the scalability of all methods by fixing either the number of users (i.e., functions) 
or the maximum budget (equal to the number of items), while varying the other.
In Figure \ref{fig:time} we demonstrate the running time of all algorithms for the task of making a synthetic playlist. 
In this case, we generate a dataset by assuming that each user likes a small random subset of items.
We generate a random budget for each user, from 1 to the given maximum budget, 
and a random cost from 1~to~10 for each item.

When comparing the running time, 
the \quality algorithm is a meaningful baseline, as it produces a ranking after a single evaluation on each item over all functions, 
i.e., $\bigO(\max\{\nV \log(\nV), \nfsms \nV\})$.
Its running time varies almost linearly as a function of the budget, 
which is in contrast to the behavior of the na\"ive greedy algorithms.
Thanks to the lazy evaluation technique \citep{leskovec2007cost}, the running time of all greedy algorithms actually grows nearly linearly in the budget.
The \greedysr algorithm is slower as it is subject to frequent function evaluations, because its greedy criterion depends on the current function values.
The running time of the \DP algorithm grows quadratically in the number of functions, which has difficulty in scaling to a very large number.
On the other hand, it scales well in the number of items, and particularly, when the budget is big, it finishes quickly as there is no large item.
The running time of all except for the \random algorithm grows linearly in the number of functions, which is inevitable if the utility of items is considered.

\begin{figure}[t]
    \centering
    \subcaptionbox{Increasing maximum budget}{
    % This file was created by tikzplotlib v0.9.8.
\begin{tikzpicture}

\definecolor{color0}{rgb}{0.12156862745098,0.466666666666667,0.705882352941177}
\definecolor{color1}{rgb}{1,0.498039215686275,0.0549019607843137}
\definecolor{color2}{rgb}{0.580392156862745,0.403921568627451,0.741176470588235}
\definecolor{color3}{rgb}{0.83921568627451,0.152941176470588,0.156862745098039}
\definecolor{color4}{rgb}{0.172549019607843,0.627450980392157,0.172549019607843}
\definecolor{color5}{rgb}{0.549019607843137,0.337254901960784,0.294117647058824}
\definecolor{color6}{rgb}{0.890196078431372,0.466666666666667,0.76078431372549}

\begin{axis}[width = 0.35\textwidth, height = 0.3\textwidth, 
minor tick style={draw=none}, 
legend style={
  legend pos=outer north east
},
ytick={1e-5, 1e-1, 1e3},
legend cell align={left},
legend style={
  fill opacity=0.8,
  draw opacity=1,
  text opacity=1,
  at={(0.5,0.09)},
  anchor=south,
  draw=white!80!black
},
log basis x={10},
log basis y={10},
tick align=outside,
tick pos=left,
x grid style={white!69.0196078431373!black},
xlabel={Maximum budget},
xmin=7.07945784384138, xmax=14125.3754462276,
xmode=log,
xtick style={color=black},
y grid style={white!69.0196078431373!black},
ylabel={Running time (s)},
ymin=9.59617302250428e-06, ymax=1239.74560273728,
ymode=log,
ytick style={color=black}
, legend style={at={(1.2,1)}, anchor=north west}]
\addplot [semithick, color0, mark=*, mark size=1.5, mark options={solid}]
table {%
10 0.0104139630000004
100 0.571955056
1000 10.108779
10000 135.019488217
};
\addlegendentry{Greedy-U}
\addplot [semithick, color1, mark=*, mark size=1.5, mark options={solid}]
table {%
10 0.0105613729999998
100 0.607377626
1000 8.674614012
10000 93.715140361
};
\addlegendentry{Greedy-W}
\addplot [semithick, color2, mark=*, mark size=1.5, mark options={solid}]
table {%
10 0.043976207
100 1.265490826
1000 14.929455466
10000 171.217683123
};
\addlegendentry{Subm}
\addplot [semithick, color3, mark=*, mark size=1.5, mark options={solid}]
table {%
10 0.03413961
100 2.645856123
1000 36.744558379
10000 530.444681095
};
\addlegendentry{AG}
\addplot [semithick, color4, mark=*, mark size=1.5, mark options={solid}]
table {%
10 0.0140123540000001
100 0.131146587
1000 1.282779214
10000 6.58408578800004
};
\addlegendentry{Quality}
\addplot [semithick, color5, mark=*, mark size=1.5, mark options={solid}]
table {%
10 2.24280000002963e-05
100 3.79649999997511e-05
1000 0.000179436000003363
10000 0.000922897000009471
};
\addlegendentry{Random}
\addplot [semithick, color6, mark=*, mark size=1.5, mark options={solid}]
table {%
10 0.0118911780000002
100 20.416468848
1000 171.197053244
10000 0.837398998
};
\addlegendentry{DP}
\legend{}\end{axis}

\end{tikzpicture}
    }
    \hskip -2ex
    \subcaptionbox{Increasing number of users}{
    % This file was created by tikzplotlib v0.9.8.
\begin{tikzpicture}

\definecolor{color0}{rgb}{0.12156862745098,0.466666666666667,0.705882352941177}
\definecolor{color1}{rgb}{1,0.498039215686275,0.0549019607843137}
\definecolor{color2}{rgb}{0.580392156862745,0.403921568627451,0.741176470588235}
\definecolor{color3}{rgb}{0.83921568627451,0.152941176470588,0.156862745098039}
\definecolor{color4}{rgb}{0.172549019607843,0.627450980392157,0.172549019607843}
\definecolor{color5}{rgb}{0.549019607843137,0.337254901960784,0.294117647058824}
\definecolor{color6}{rgb}{0.890196078431372,0.466666666666667,0.76078431372549}

\begin{axis}[width = 0.35\textwidth, height = 0.3\textwidth, 
minor tick style={draw=none}, 
legend style={
  legend pos=outer north east
},
ytick={1e-5, 1e-1, 1e3},
legend cell align={left},
legend style={
  fill opacity=0.8,
  draw opacity=1,
  text opacity=1,
  at={(0.03,0.97)},
  anchor=north west,
  draw=white!80!black
},
log basis x={10},
log basis y={10},
tick align=outside,
tick pos=left,
x grid style={white!69.0196078431373!black},
xlabel={No. of users},
xmin=7.07945784384138, xmax=14125.3754462276,
xmode=log,
xtick style={color=black},
y grid style={white!69.0196078431373!black},
ylabel={Running time (s)},
ymin=1.51781135851901e-05, ymax=4213.34016357291,
ymode=log,
ytick style={color=black}
, legend style={at={(1.2,1)}, anchor=north west}]
\addplot [semithick, color0, mark=*, mark size=1.5, mark options={solid}]
table {%
10 0.0502327759999996
100 0.627029594
1000 8.224133119
10000 83.391840628
};
\addlegendentry{Greedy-U}
\addplot [semithick, color1, mark=*, mark size=1.5, mark options={solid}]
table {%
10 0.0478523379999998
100 0.551392582
1000 8.613276793
10000 87.559526716
};
\addlegendentry{Greedy-W}
\addplot [semithick, color2, mark=*, mark size=1.5, mark options={solid}]
table {%
10 0.0928983840000002
100 1.276945174
1000 14.70474761
10000 160.100372764
};
\addlegendentry{Subm}
\addplot [semithick, color3, mark=*, mark size=1.5, mark options={solid}]
table {%
10 0.150920381
100 2.618953899
1000 28.289342302
10000 298.566623363
};
\addlegendentry{AG}
\addplot [semithick, color4, mark=*, mark size=1.5, mark options={solid}]
table {%
10 0.013993299
100 0.140254461
1000 1.340758859
10000 13.681021741
};
\addlegendentry{Quality}
\addplot [semithick, color5, mark=*, mark size=1.5, mark options={solid}]
table {%
10 4.00859999998282e-05
100 3.67289999996245e-05
1000 4.02539999981855e-05
10000 4.98390000984728e-05
};
\addlegendentry{Random}
\addplot [semithick, color6, mark=*, mark size=1.5, mark options={solid}]
table {%
10 0.132958667
100 12.836610665
1000 1741.146112783
10000 nan
};
\addlegendentry{DP}
%\legend{}
\end{axis}

\end{tikzpicture}
    }
    \caption{\label{fig:time}Running time of all methods for the task of making a synthetic playlist.
    }
\end{figure}
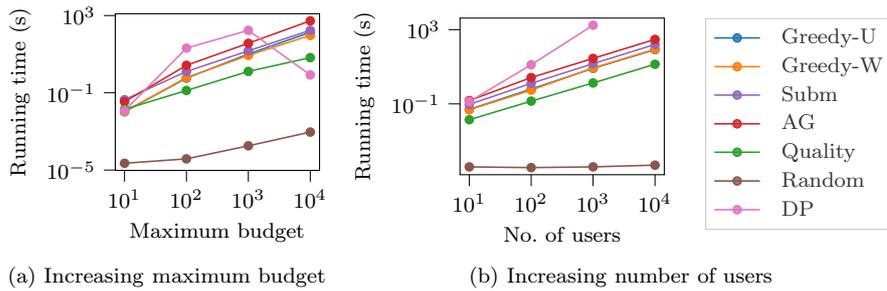

\section{Conclusions}
\label{section:conclusion}

In this paper, we introduce a novel problem in the active area of submodular optimization. 
Our problem, {max-submodular ranking} (\msr),  
ask to find a ranking of items such that the sum of multiple budgeted submodular utility is maximized.
The \msr problem has wide application in the ranking of web pages, ads, and other types of items.
We propose several practical algorithms with approximation guarantees for the \msr problem, 
with either cardinality or knapsack budget constraints.
We empirically demonstrate the superior performance of the proposed algorithms on real-life datasets, 
compared with a state-of-the-art baseline and other meaningful heuristics.

One direction for future work is to narrow the gap between the approximation ratio and the lower bound.
Another direction is to study the online version of the \msr problem, 
to allow for the arrival of new submodular functions.
Other potential directions include imposing a more general constraint for each submodular function and experimenting with new applications.

\begin{acknowledgements}
This research is supported by the Academy of Finland projects MALSOME (343045), AIDA (317085) and MLDB (325117),
the ERC Advanced Grant REBOUND (834862), 
the EC H2020 RIA project SoBigData++ (871042), 
and the Wallenberg AI, Autonomous Systems and Software Program (WASP) 
funded by the Knut and Alice Wallenberg Foundation.
\end{acknowledgements}


\bibliographystyle{spbasic}
\bibliography{references}

\end{document}